\documentclass[12pt,draftclsnofoot,onecolumn]{IEEEtran}

\usepackage{amsmath}
\usepackage{graphicx,psfrag,epsf}
\usepackage{enumerate}
\usepackage{amssymb,amsfonts,mathtools}
\usepackage{bm}
\usepackage{color, colortbl}
\usepackage{algorithm}
\usepackage{algorithmic}
\usepackage{balance,cite}
\usepackage{multirow}
\usepackage{wrapfig,bbm}

\usepackage{hyperref}
\usepackage{booktabs}

\usepackage{amsthm}
\newtheoremstyle{remboldstyle}
  {}{}{\itshape}{}{\bfseries}{.}{.5em}{{\thmname{#1 }}{\thmnumber{#2}}{\thmnote{ (#3)}}}
\theoremstyle{remboldstyle}

\newcommand{\rom}[1]{\uppercase\expandafter{\romannumeral #1\relax}}
\DeclarePairedDelimiterX{\norm}[1]{\lVert}{\rVert}{#1}
\DeclarePairedDelimiterX{\bnorm}[1]{\biggl\lVert}{\biggr\rVert}{#1}
\DeclarePairedDelimiterX{\abs}[1]{\lvert}{\rvert}{#1}

\newtheorem{definition}{Definition}
\newtheorem{remark}{Remark}
\newtheorem{theorem}{Theorem}

\newtheorem{lemma}{Lemma} 

\newtheorem{example}{Example}

\newtheorem{proposition}{Proposition}

\def\R{\mathbb{R}}
\def\T{{ \mathrm{\scriptscriptstyle T} }}
\def\P{{ \mathrm{pr} }}
\def\v{{\varepsilon}}

\def\F{\textsc{F}}
\def\E{\mathbb{E}}

\def\Y{\mathcal{Y}}
\def\Z{\mathcal{Z}}
\def\F{\mathcal{F}}
\def\P{\mathbb{P}}
\def\M{\mathcal{M}}

\def\de{\overset{\Delta}{=}}

\def\limp{\rightarrow_{p}}

\def\card{\textrm{card}}
\def\de{\overset{\Delta}{=}}

\def\v{\varepsilon}
\def\Fv{F_{\v}}
\def\G{G_{\v}}

\def\var{var}
\def\i{\mathbbm{1}} 
\def\B{B} 
\def\YY{\mathbb{Y}}
\def\ZZ{\mathbb{Z}}

\def\Fy{F_Y} 
\def\C{A} 

\def\Q{U} 
\def\q{u} 
\def\i{\mathbbm{1}} 
\def\a{\textrm{essinf } \Y} 
\def\b{\textrm{esssup } \Y} 
\def\T{T} 
\def\Ra{R} 
\def\ra{r} 
\def\card{\textrm{card}}
\def\rand{W}
\def\Bern{\textrm{Bern}}

\begin{document}

\title{Interval Privacy: A Framework for Privacy-Preserving Data Collection}

 \author{Jie~Ding and 
 		Bangjun~Ding
 \thanks{J.~Ding is with the School of Statistics, University of Minnesota Twin Cities, MN 55414, USA. B.~Ding is with the School of Statistics and Finance, East China Normal University, Shanghai, China. 
 } 
 \thanks{
 This paper has an associated open-source project actively maintained at \url{https://github.com/JieGroup/IP}.}
 }

\markboth{IEEE Transactions on Signal Processing 
}
{Shell \MakeLowercase{\textit{et al.}}: IEEE Transactions on Signal Processing}

\maketitle

\begin{abstract}
The emerging public awareness and government regulations of data privacy motivate new paradigms of collecting and analyzing data that are transparent and acceptable to data owners. We present a new concept of privacy and corresponding data formats, mechanisms, and theories for privatizing data during data collection. The privacy, named Interval Privacy, enforces the raw data conditional distribution on the privatized data to be the same as its unconditional distribution over a nontrivial support set. Correspondingly, the proposed privacy mechanism will record each data value as a random interval (or, more generally, a range) containing it. The proposed interval privacy mechanisms can be easily deployed through survey-based data collection interfaces, e.g., by asking a respondent whether its data value is within a randomly generated range. Another unique feature of interval mechanisms is that they obfuscate the truth but do not perturb it. Using narrowed range to convey information is complementary to the popular paradigm of perturbing data. Also, the interval mechanisms can generate progressively refined information at the discretion of individuals, naturally leading to privacy-adaptive data collection. We develop different aspects of theory such as composition, robustness, distribution estimation, and regression learning from interval-valued data. Interval privacy provides a new perspective of human-centric data privacy where individuals have a perceptible, transparent, and simple way of sharing sensitive data. 
\end{abstract}

\begin{IEEEkeywords}
data collection,
human-computer interface,
interval data,
interval privacy,
interval mechanism,
local privacy,
privacy,
survey.
\end{IEEEkeywords}

\clearpage
\tableofcontents
\clearpage

\section{Introduction}
\label{sec_intro}

With new and far-reaching laws such as the General Data Protection Regulation~\cite{voigt2017eu} and frequent headlines of large-scale data breaches, there has been a growing societal concern about how personal data are collected and used~\cite{evans2015biometrics,cross2020privacy}.
Consequently, data privacy has been an increasingly important factor in designing signal processing and machine learning services.
This paper will address the following scenario often seen in practice. 
Suppose that Alice is the agent who creates and holds raw data, which will be collected by another agent Bob. 
On the one hand, Alice may not trust Bob or the transmission channel to Bob. On the other hand, Bob is interested in population-wide inference using statistics provided by Alice and many other individuals, but not necessarily the exact value of Alice. 

The above learning scenario is quite common in, e.g., Machine-Learning-as-a-Service cloud services~\cite{ribeiro2015mlaas,DingIL}, multi-organizational Assisted Learning~\cite{DingAssist,DingGAL,DingAssistSGD}, survey-based inferences~\cite{couper2001web,litwin2003assess}, and information fusion~\cite{wang2018fusing,sun2019relationship,DingCollabParam}. 
The formalization of individual-level data privacy and population-level estimation utility has motivated active research on what is generally referred to as local data privacy across fields such as data mining~\cite{evfimievski2003limiting}, security~\cite{kasiviswanathan2011can}, statistics~\cite{duchi2018minimax}, and information theory~\cite{du2012privacy,sun2016towards}. 
The general goal of local data privacy is to suitably randomize raw data during the data collection and evaluate it through an appropriate framework.

In this work, we propose a notion of local privacy named \textit{interval privacy} for protecting data collected for further inferences. The main idea is to enforce privacy in such a way that the distribution of the raw data conditional on its privatized data remains the same (up to a normalizing constant) on a moderately large support set.  
In other words, no additional information is gained except that the support of the data becomes narrow.
Accompanying the notion of interval privacy, we use the size (in a measure-theoretic sense) of the conditional support to quantify the level of privacy. The size, named \textit{privacy coverage}, enables a natural interpretation and perception of the amount of ambiguity exposed to the data collector.
We then introduce interval privacy mechanisms for realizing data collection in practice. 

Our perspective of privacy is motivated by the following practical concerns. Suppose that an organization collects privacy-sensitive information from individuals, e.g., an organization gathers users' demographic information. A concern is \textit{how to develop a data collection interface so that individuals can easily perceive that the collected data are at their discretion}. In other words, individuals do not have to submit exact raw data first and then rely on any subsequent processing of those data, which can be a black-box procedure obscure to the public.
The main idea of interval mechanisms is to generate random intervals that partition the data domain and collect the interval containing the underlying data. It can be naturally implemented as a transparent yet simple survey interface, where an individual can directly see the ultimate collection and perceive its ambiguity. 
As individuals may have different privacy sensitivities, another related concern is \textit{how to obtain data in a way adaptive to individual-level privacy}. Interval privacy addresses this by progressively collecting data from wider (and thus more private) intervals to narrower ones, meaning that individuals may respond, not respond, or respond further at their discretion.

Our notion of privacy naturally leads to a new form of disclosing and collecting sensitive data, namely representing them as intervals instead of points. 
For example, a sensor's accurate distance with the target $y=10$ (in meters) is privatized by first generating a random threshold, say $20$, and then publicizing the corresponding interval $(-\infty, 20]$; or an individual's $60k$ salary (in dollars) is privatized by first generating random thresholds, say $41k$ and $85k$, and then reporting the interval $(41k, 85k]$. 
The random thresholds can be generated from any distribution known to the data collector, e.g., Gaussian, Logistic, and Uniform distributions, independent of the underlying data. It is worth noting that an interval is not necessarily symmetric around the underlying raw data value. 
Tab.~\ref{fig_snap} illustrates $Y$ and its private counterpart in a dataset that we will revisit in experimental studies. Each individual's privacy coverage describes the interval size or level of ambiguity. For example, the data with $97\%$ coverage is less private than the $99\%$ one, which is in line with the perception that $Y<82.2$ reveals more information than $Y<85.4$.
We will show several fundamental properties of the proposed privacy mechanism to render its broad applicability. These include the composition property that characterizes the level of overall privacy degradation in multiple queries to the same data, robustness to pre-processing, robustness to post-processing,  distributional identifiability, and extensions from intervals to general ranges.
We will demonstrate interval-private data for several inference tasks, including moment estimation, functional estimation, and supervised regression.

\begin{table}[tb]
\centering
\caption{A snapshot of the `life expectancy' database~\cite{lifeData} to be studied in Subsection~\ref{subsec_data2}. The life expectancy, $Y$, is privatized into random intervals, with an overall privacy coverage of $60.3\%$. 
}
\label{fig_snap}
\scalebox{1}{
\begin{tabular}{ccccc}
\toprule
$Y$ (in years)              & $59.3$      & $82.3$      & $79.5$           & $51.7$         \\ \midrule
Privatized $Y$   & $(0, 82.2]$ & $(0, 85.4]$ & $(64.2, \infty)$ & $(46.9, 72]$ \\ \midrule
Privacy coverage (using $\P_Y$) & $97\%$      & $99\%$      & $70\%$           & $47\%$ \\ \bottomrule      
\end{tabular}
}
\end{table}

\begin{figure}[tb]
\begin{center}
\centerline{\includegraphics[width=0.8\columnwidth]{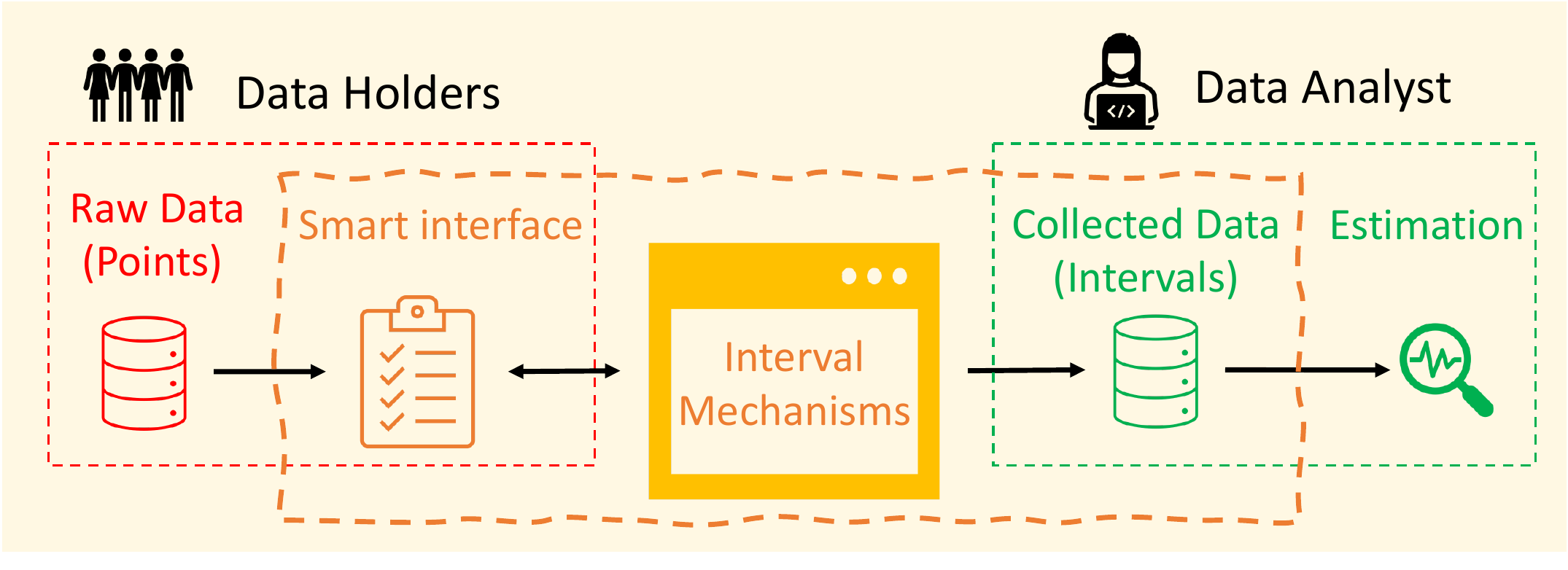}}
\vskip -0.1in
\caption{A generic data collection system based on interval privacy.}
\label{fig_genericSystem}
\end{center}
\vskip -0.1in
\end{figure}

The main contributions of this paper are summarized below. 
\begin{itemize}
	\item 
	We develop a new perspective of data privacy named interval privacy, particularly suitable for privacy-sensitive data collection. We develop interval mechanisms and show their desirable interpretations to implement interval privacy naturally. Fig.~\ref{fig_genericSystem} illustrates a general use scenario where individuals' private data are collected through an interface that obfuscates each data point into an interval (or, in general, a range). We show several unique features of an interval privacy mechanism. First, it tells the truth while obfuscating the truth, which is important for some application domains such as census and defense scenarios where information needs to be correct. Second, it can be easily deployed through survey systems, with an interpretable and perceptible human-computer interface. Third, such an interface can allow progressive narrowing of collected intervals and thus be adaptive to individuals' privacy sensitivities that are likely to vary in practice. Fig.~\ref{fig_surveySystem} illustrates a general survey system built upon interval privacy, which, unlike conventional surveys widely used in various fields such as sociology, political science, and psychometrics~\cite{couper2001web,litwin2003assess}, generates questions in an individual-specific and data-adaptive manner.
	To our best knowledge, this is the first work that advocates the use of random ranges for privacy-preserving data collection and the use of randomly generated questions in survey designs.
	
	\item 
	We develop fundamental properties of the proposed privacy mechanism, including the composition property that characterizes privacy leakage under multiple queries of the same data, the robustness to pre-processing and post-processing, and the identifiability of the underlying data distributions.
	We exemplify the use of interval privacy in estimating population distribution, statistical functional, and regression function, and show that the data collector does not necessarily need to know the distributional form of raw data for accurate population-level inference.
	In particular, we provide a general theory to show the topology and probabilistic structures needed to reconstruct the population distribution from random ranges non-parametrically.
	We develop several extensions to address individual-level privacy guarantees.
	We also develop a general method to perform supervised regression with interval-privatized responses. The technique can be applied to various interval-private data types, including pure intervals or a mixture of intervals and points. 

	\item 
	We experimentally demonstrate the proposed concepts, data formats, properties, and methods. We also discuss the connections between interval privacy and the existing literature from multiple angles. For example, we will point out (in the supplementary document) that interval privacy is neither a generalization nor a specialization of (local) differential privacy~\cite{evfimievski2003limiting,dwork2006calibrating,kasiviswanathan2011can}. 
\end{itemize}

\begin{figure}[tb]
\begin{center}
\centerline{\includegraphics[width=0.9\columnwidth]{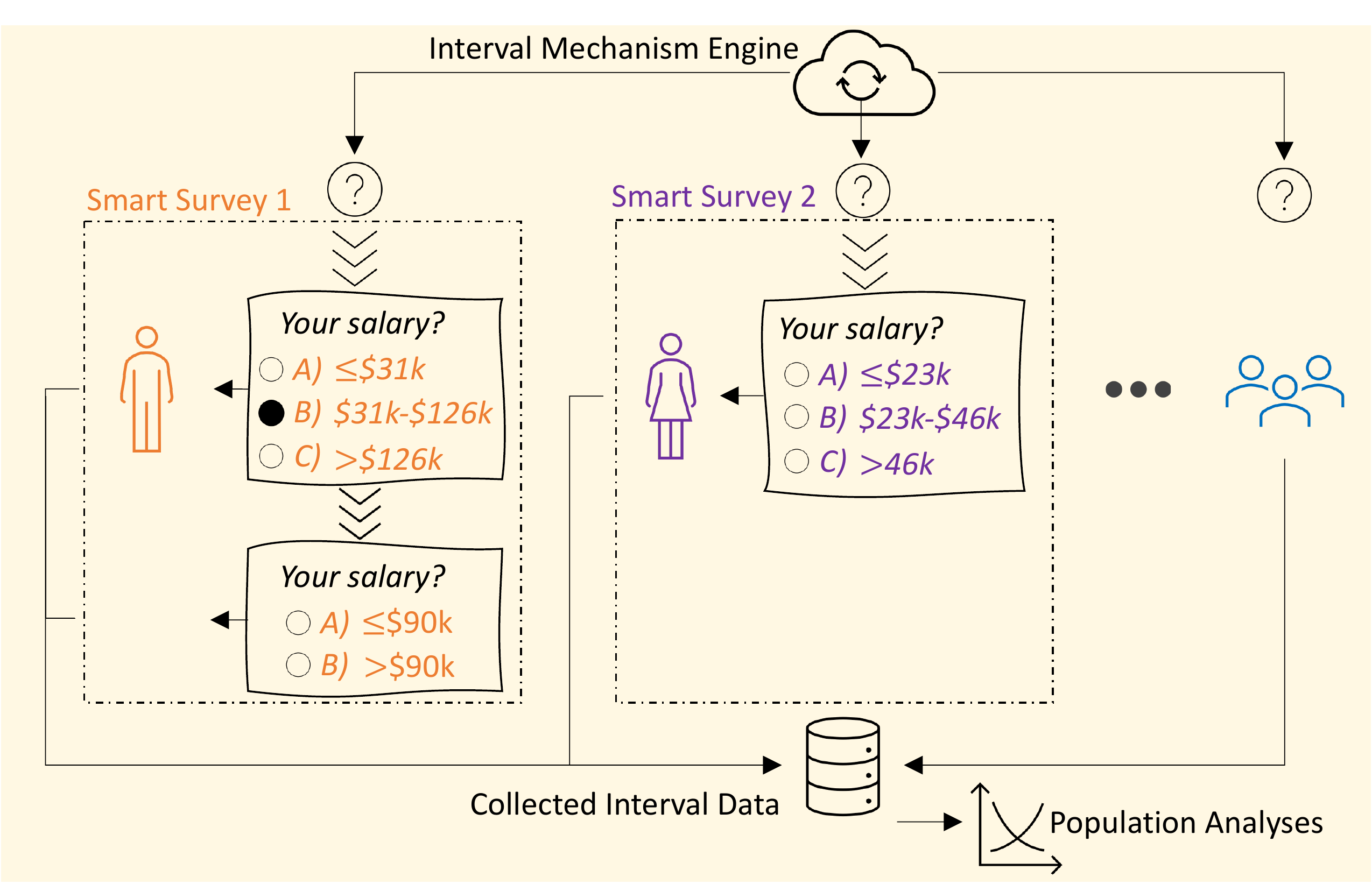}}
\vskip -0.1in
\caption{Illustration of a general survey system that generates different human-computer interfaces for participating individuals.}
\label{fig_surveySystem}
\end{center}
\vskip -0.2in
\end{figure}

The rest of the paper is outlined below. 
In Section~\ref{sec_form}, we introduce the basic concept of interval privacy and use simple examples to explain its use scenarios. 
In Section~\ref{sec_general_foundation}, we introduce general interval privacy mechanisms, data formats, theoretical properties, and various practical implications. 
In Section~\ref{sec_exp}, we provide experimental studies.
In Section~\ref{sec_literature}, we further discuss some related literature. 
We conclude the paper in Section~\ref{sec_con} and include proofs in the Appendix. Additional discussions and details are in the supplementary document.

\section{Interval Privacy} \label{sec_form}

\subsection{Notation}
We let $Y$ denote a continuously-valued random variable representing the raw data throughout the paper.
Suppose that the raw data $Y_1,\ldots,Y_n \in \Y\subset \R$ are i.i.d. with probability $\P_Y$, density $p_Y$, and cumulative distribution function (CDF) $\Fy$. 
We will write $i=1,\ldots,n$ as $i \in [1:n]$. 
For a random vector $Q$, we let $Q^{(i)}$ and $Q_i$ denote its $i$-th entry and $i$-th observation, respectively, unless otherwise stated. 

We consider the local data privacy scenario where there are many data owners and one data collector. A data owner is an individual that holds a private data value (represented by $Y$) and does not trust the data collector. A data collector's genuine goal is to infer distributional information of $Y$ instead of each individual's data value.
As such, a general local privacy scheme uses a random mechanism $\M$ that maps each $Y$ to another variable $Z \in \Z$ and then collects $Z$. The mechanism is often represented by a conditional distribution of $Z \mid Y$. 
The random variable $Z$ may be constructed by a measurable function of $Y$, or a function of $Y$ and other auxiliary random variables.
We assume that the joint distribution of $[Y, Z]$ exists and has a density with respect to the Lebesgue measure.
Suppose that $S$ is a Borel set. We let $L(S) \de \P_Y(S)$ denote the `size' of $S$ (which remains the same throughout the paper).

\subsection{Interval Privacy} \label{subsec_IP_def}

\begin{definition}[Interval Privacy]\label{def_IP}
  A mechanism $\M$ has the property of interval privacy if almost surely for all $y_1,y_2 \in S_z$,  
  \begin{align}
    \frac{p_{Y \mid Z}(y_1 \mid Z=z)}{p_{Y \mid Z}(y_2 \mid Z=z)} 
    = \frac{p_{Y}(y_1)}{p_{Y}(y_2)}\label{eq_def},
  \end{align}
  where $p_{Y \mid Z}$ denotes the distribution of $Y$ conditional on $Z$ and $S_z$ is the support of $Y$ given $Z=z$.

  The \textit{privacy coverage} of $\M$, denoted by $\tau(\M)$, is defined by $\E(L(S_{Z}))$, where the expectation is over $Z$, and $L$ is the size under the prior distribution of $Y$ (namely $p_Y$).
  An $\M$ is said to have $\tau$-interval privacy if $\tau(\M)\geq \tau$.
\end{definition}

\textit{Implication 1}:
Equation~(\ref{eq_def}) means that the conditioning on $Z=z$ does not provide extra information except that $y$ falls into $S_z$. If $y_1 \neq y_2$, and they fall into the same support $S_z$, their likelihood ratio remains the same as if no action were taken.
Equation~(\ref{eq_def}) also implies that 
\begin{align}
  p_{Y \mid Z}(y \mid Z=z) = c_z \i_{y \in S_z} \cdot p_Y(y) \label{eq15}
\end{align}
holds for the normalizing constant $c_z=1/\int_{S_z} p_Y(y) dy$.

\textit{Implication 2}:
Suppose that $Y=y_1$ is to be protected. Interval privacy creates ambiguity by obfuscating the observer with sufficiently many $y_2$'s in a neighborhood whose posterior ratios do not vary by incorporating the new information $Z=z$.
Also, suppose that $S_z$ is a (closed or open) interval, then the finite cover theorem implies the following alternative to the above second condition.
For all $y$ in the interior of $S_z$, there exists an open neighborhood of $y$, $U(y) \subset S_z$, where (\ref{eq_def}) holds for all $y_1,y_2 \in U(y)$.

\textit{Implication 3}: By its definition, the privacy coverage $\tau(\M)$ takes values from $[0,1]$. The privacy coverage quantifies the average amount of ambiguity or the level of privacy. A larger value indicates increased privacy. Likewise, for each (nonrandom) raw-privatized data pair, $(y, z)$, we introduce $L(S_z)$ as the \textit{individual privacy coverage}, interpreted as the privacy level for a particular data item being collected (illustrated in the third row of Tab.~\ref{fig_snap}). 

A related measure is $1-\tau(\M)$ which naturally describes the \textit{privacy leakage}. For instance, the coverage of $\Y$ is one, and the leakage is zero, meaning no privacy is leaked; Meanwhile, the coverage of $y$ (as a degenerate interval) is zero.
To realize interval privacy, we will introduce natural interval mechanisms that convert $y$ to a random interval that contains $y$. 
For example, $S_z$ is in the form of $(-\infty, u]$ or $(u,\infty)$, encoded by the vector $z=[u, \i_{y \leq u}]$.

The notion of interval privacy appears to be related to information privacy~\cite{du2012privacy,sun2016towards} that requires the posterior-prior density ratio $p_{Y \mid Z}(y \mid Z=z)/p_Y(y)$ to stay in $[e^{-\alpha},e^{\alpha}]$ for all feasible $y$ and $z$ under a constant (privacy budget) $\alpha>0$. Nevertheless, interval privacy and information privacy do not imply each other. In fact, by its definition, $\alpha$-information privacy implies $2\alpha$-local differential privacy, which coincides with interval privacy only when $\alpha=0$, the trivial case that $Y$ and $Z$ are independent (elaborated in the supplement). In this regard, interval privacy provides a unique angle of privatizing information complementary to the existing notions. 

We provide Fig.~\ref{fig_V1_view} to visualize our unique approach to protecting data information. 
It shows the data format of interval data and its released information of the raw data as implied by posterior uncertainty. It also visualizes the popular approach that privatizes data by perturbations. From a Bayesian perspective, the perturbation changes the density shape, while the interval approach changes the essential support.  
In the plot, we generated raw data $y$ from the standard Gaussian. The interval privacy used the standard Logistic random variable $U$ and reports either `$\leq U$' or `$>U$,' resulting in around $0.25$ privacy leakage. The perturbation approach truncated the raw data within $[-3,3]$ and added the Laplacian noise so that it achieves a $2$-local differential privacy.  

\begin{figure}[tb]
\begin{center}
\centerline{\includegraphics[width=0.7\columnwidth]{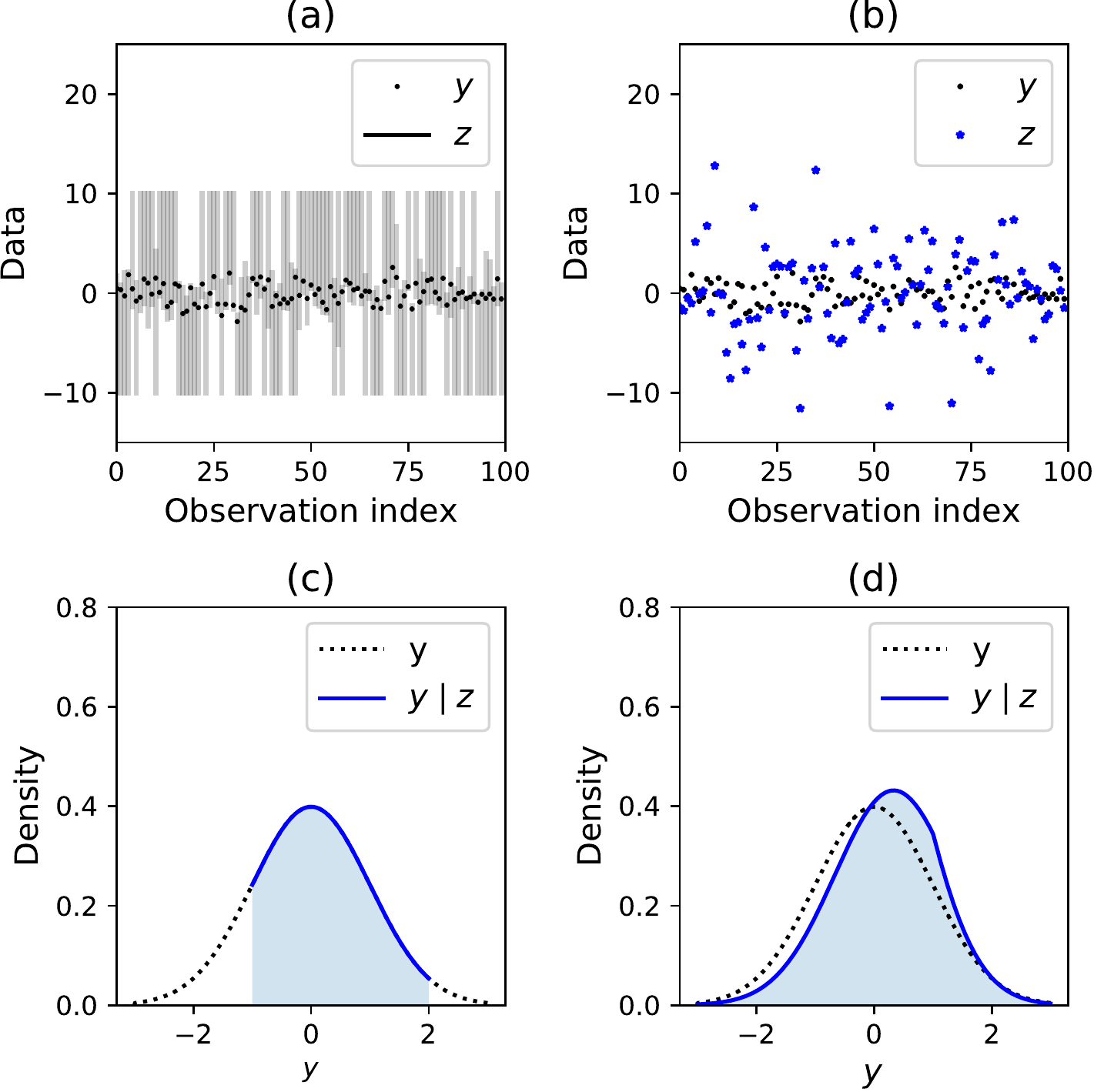}}
\vskip -0.15in
\caption{An illustration of interval privacy (left column) and local differential privacy with noise perturbation (right column) in terms of: the raw data $y$ and interval-privatized data (left-up), the posterior distribution of $Y$ given an interval $[-1,2]$ (left-bottom), raw data $y$ and Laplacian-perturbed (point) data (right-up), and the posterior of $Y$ given an observed point $2$ (right-bottom).}
\label{fig_V1_view}
\end{center}
\vskip -0.15in
\end{figure}

\subsection{Explanations of Interval Privacy via Simple Examples} \label{subsec_example}

This section provides simple examples of data formats, mechanisms, and practical implications regarding interval privacy. 
We will introduce formal definitions of different mechanisms and theoretical foundations in Section~\ref{sec_general_foundation}. 

Suppose that a data analyst aims to study the population distribution of salary. To collect the salary information from an individual (say Alice) without revealing the underlying value, Alice is asked to report whether the salary is above a threshold or not. This is illustrated in Fig.~\ref{fig_example_age}(a). This naturally leads to the following privacy mechanism, perhaps the simplest interval mechanism. Only the indicator of whether the salary is larger than a randomly generated threshold is reported.

\begin{figure}[tb]
\begin{center}
\centerline{\includegraphics[width=0.6\columnwidth]{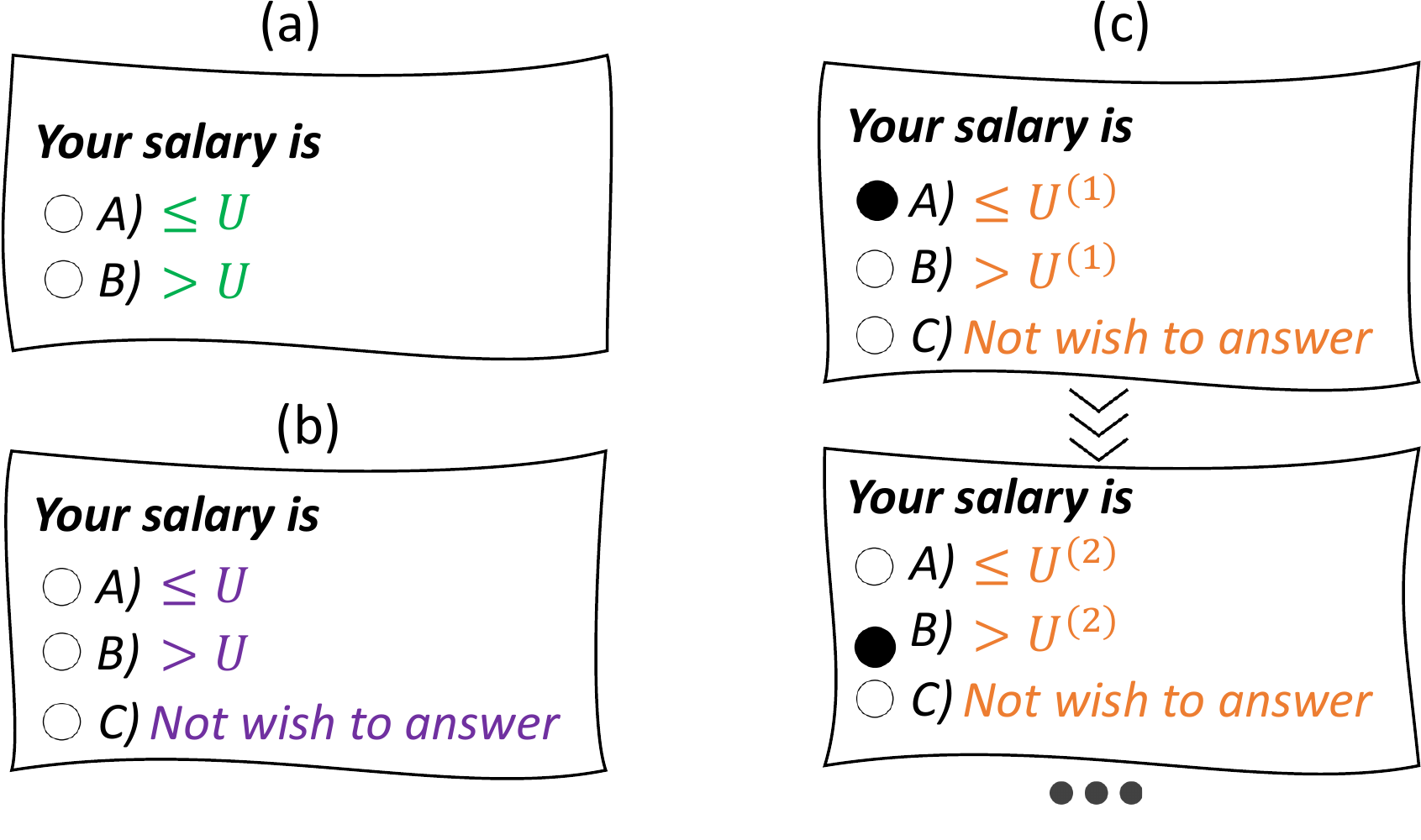}}
\vskip -0.2in
\caption{Single-choice questions that collect interval-private data in different ways, including a) `mandatory', b) `opt-in', and c) `progressive opt-in'.}
\label{fig_example_age}
\end{center}
\vskip -0.1in
\end{figure}

\vspace{1mm}
\noindent 
$\bullet$ \textit{Case-I interval mechanism}:
Let $U \in \R$ be a random variable independent with $Y$, referred to as an anchor point. 
Either $Y \leq U$ or $Y > U$ is observed.
The observations are $n$ i.i.d. copies of $Z=[U,\Delta]$, where $\Delta = \i_{Y \leq U}$ is an indicator variable.

Likewise, we also define the following mechanism that admits a bounded interval (e.g., $\$60$-$80k$).
More general mechanisms will be introduced in Section~\ref{sec_general_foundation}. 

\vspace{1mm}
\noindent
$\bullet$ \textit{Case-II interval mechanism}:
Let $[U,V] \in \R^2$ be a random variable that satisfies $\P(U\leq V)=1$ and is independent with $Y$. 
Either $Y \leq U$, $U < Y \leq V$, or $Y > V$ is observed.
The observations are $n$ i.i.d. copies of $Z=[U,V,\Delta, \Gamma]$, where  
$\Delta = \i_{Y \leq U}$ and $\Gamma = \i_{U < Y \leq V}$ are indicator variables.

We summarize some features of interval mechanisms below.

\indent 
\textbf{1) Conditional non-informativeness}: We will show in Subsection~\ref{subsec_canonical_mechanism} that  the collected data $Z$ in the above examples satisfy the interval privacy (Definition~\ref{def_IP}). Thus, conditional on the revealed support set, e.g., $(-\infty, U]$, no additional information is revealed since the relative probability densities of $Y$ conditional on $Y \leq U$ do not differ from unconditional ones. So, the only information provided by $Z$ about the raw data $Y$ is an (often wide) range that contains $Y$. 

\indent 
\textbf{2) Information fidelity}: 
An interesting aspect of the interval privacy mechanism is that it collects obfuscated data instead of perturbed data.  
Here, we use the term `obfuscation' to refer to the process $Y \rightarrow Z$ that any deductive reasoning based on $Z$ does not contradict the truth of $Y$, referred to as information fidelity. In contrast, `perturbation' means one cannot make a factual statement from observing $Z$. Both the terms are materialized by introducing randomness (but in different ways). We will theoretically elaborate on their difference in the supplement.
Maintaining information fidelity is vital in many applications domains such as census, security, and defense, where collecting perturbed data can lead to misinterpretations or disastrous decisions. The interval-private data convey information without lying about the underlying values. As we will show later, even if each point is obfuscated into a fairly wide range, one can still reconstruct the underlying population distribution without systematic biases. 

The obfuscation process of interval privacy offers another practical benefit. Consider scenarios where a resourceful organization collects private information from anonymized individuals. Individuals hope to easily perceive that the already-collected data are indeed private. Existing privacy schemes such as homomorphic encryption~\cite{gentry2009fully} and local differential privacy~\cite{evfimievski2003limiting} often need the collecting organization to implement sophisticated cryptography- or randomization-based procedure at the backend. Consequently, their privacy architectures may require individuals to submit exact raw data in the collecting interface, which inevitably raises trustworthiness issues. A potential remedy is to apply privatization immediately after data collection and publicize the source codes. But even in that case, it may not be transparent to individuals (especially to the public). 
In contrast, an organization can transparently deploy the proposed interval privacy mechanisms through electronic survey-based data collection infrastructures. Such a privacy interface allows an individual to perceive the level of privacy directly and at peace.

\indent 
\textbf{3) Distributional identifiability}: 
It is worth noting that generating random $U$ in the above Case-I mechanism is essential. If $U$ is deterministic, it is impossible to accurately estimate the distribution of $Y$ since one can always find a distinct distribution whose mass on the pre-determined intervals coincides.
Suppose that the essential support of $U$ contains that of $Y$. It has been shown under reasonable conditions that the distribution of $Y$ can be consistently estimated from interval observations even if the underlying distribution is not parameterized~\cite{groeneboom1992information}. 
We will revisit the nonparametric estimation method and develop a new theory for general interval mechanisms in Subsection~\ref{subsec_topology}.
To illustrate distributional identifiability, we generate $1000$ points of $Y$ from a standard Logistic distribution and Case-I interval-private data from $U$ that follows a Logistic distribution whose scale is $2$. Fig.~\ref{fig_cdf} (left plot) shows parametric and nonparametric estimations of the CDF $\Fy$ from the interval data.
The parametric estimation uses the standard maximum likelihood approach. 
The nonparametric estimation uses the self-consistency algorithm~\cite{turnbull1976empirical} implemented in the `Icens' R package~\cite{gentleman2010icens}.

\begin{example}[Functional Estimation]\label{example1}
  Suppose that an analyst is interested in estimating a smooth functional $K(\Fy)$ of the underlying distribution function $\Fy$. A nonparametric estimator is $K(\hat{F}_Y)$ where $\hat{F}_Y$ is the nonparametric maximum likelihood estimator of $\Fy$~\cite{groeneboom1992information}. 
  Specifically, all moment functionals $K: \Fy \mapsto \int_{\Y}y^k d\Fy(y) $, or more generally,
  linear functionals in the form of 
	$K: \Fy \mapsto \int_{\Y} \phi(y) d\Fy(y)$
   can be estimated in this way. 
\end{example}

\begin{example}[Mean Estimation]\label{eg2}
	Sometimes, a statistical functional may be directly estimated without the need of estimating $\Fy$. For example, suppose that the raw data are i.i.d. $Y_i \in [a,b]$ for $i\in [1:n]$,  with unknown mean $\mu$. The observations are $Z_i=[U_i,\Delta_i]$, $i\in [1:n]$, from the Case-I mechanism with $U_i \sim_{i.i.d.} \textrm{Uniform}[a,b]$.  
	We provide the following estimator and will show that it is a $\sqrt{n}$-consistent and unbiased estimator of $\mu$. 
	\begin{align}
		\hat{\mu}_n = \frac{1}{n} \sum_{i=1}^n \biggl( \Delta_i (2U_i-b) + (1-\Delta_i) (2U_i-a)  \biggr)	. \label{eq_eg2}
	\end{align}
\end{example}

\begin{proposition} \label{prop_eg2}
	The estimator in Example~\ref{eg2} satisfies $\E(\hat{\mu}_n) = \mu$ and $\var(\hat{\mu}_n) = O(n^{-1})$.	
\end{proposition}

\indent 
\textbf{4) Achievability}:  The ambiguity as quantified by privacy coverage $\tau$  (in Definition~\ref{def_IP}) can be controlled by the distribution of $[U, V]$, or the number of intervals, e.g., two in Case-I and three in Case-II. The larger $\tau$, the more ambiguity and thus more protection. The following result shows that any privacy coverage in $[c,1]$ for a constant $c\in [0,1)$ is achievable.

\begin{theorem}[Achievability] \label{thm_cover}
  Assume that the density function of $Y$ is bounded. 
  For any $\tau \in (1/2,1)$ (respectively $(1/3,1)$) there exists a Case-I (respectively Case-II) mechanism $\M$ whose  privacy coverage is exactly $\tau$.
\end{theorem}

The result implies that a privacy mechanism exists for arbitrarily close to one privacy coverage. Also, the proof indicates that the choice is not unique. As a by-product of the proof, $\tau(\M) =  n^{-1}\sum_{i=1}^n [ \Fy(u_i)^2 +(1-\Fy(u_i))^2 ]$ is a consistent estimator of the privacy coverage for Case-I mechanisms. The estimator can be similarly extended for other mechanisms. Although the above result indicates that the max privacy near one is achievable, we may not do so in practice since there is an inherent tradeoff between privacy and estimation accuracy. To see that, we provide an example inference task below, which is interesting in its own right. 
 
 For any functional that is differentiable along Hellinger differentiable paths of distributions (e.g., linear functionals), one can derive an H\'{a}jek-LeCam convolution theorem type information lower bound, giving the best possible limit variance that can be attained under $\sqrt{n}$ convergence rate where $n$ denotes the data size~\cite{geskus1999asymptotically}. The distribution of anchor points is said to be optimal if such information lower bound is attained by the produced interval data.
 
\begin{theorem}[Optimal Anchor] \label{thm_optimalU}
  An optimal distribution of $\Q$ (in the Case-I mechanism) for estimating any linear functional in Example~\ref{example1} exists, and it has the density 
  \begin{align}
    g_{\Q}(u) = c_{\phi} \bigl\{\Fy(u)(1-\Fy(u)) \bigr\}^{1/2} \bigl|\frac{d }{d u}\phi(u)\bigr| \nonumber
  \end{align}
  if it is integrable, 
  where $c_{\phi}$ is a normalizing constant. 
\end{theorem}
The above result indicates a tradeoff between privacy coverage and statistical efficiency (in inference). Fig.~\ref{fig_cdf} exemplifies the estimation of $\Fy$ and optimal Case-I interval mechanisms.

\begin{figure}[ht]
\begin{center}
\centerline{\includegraphics[width=0.8\columnwidth]{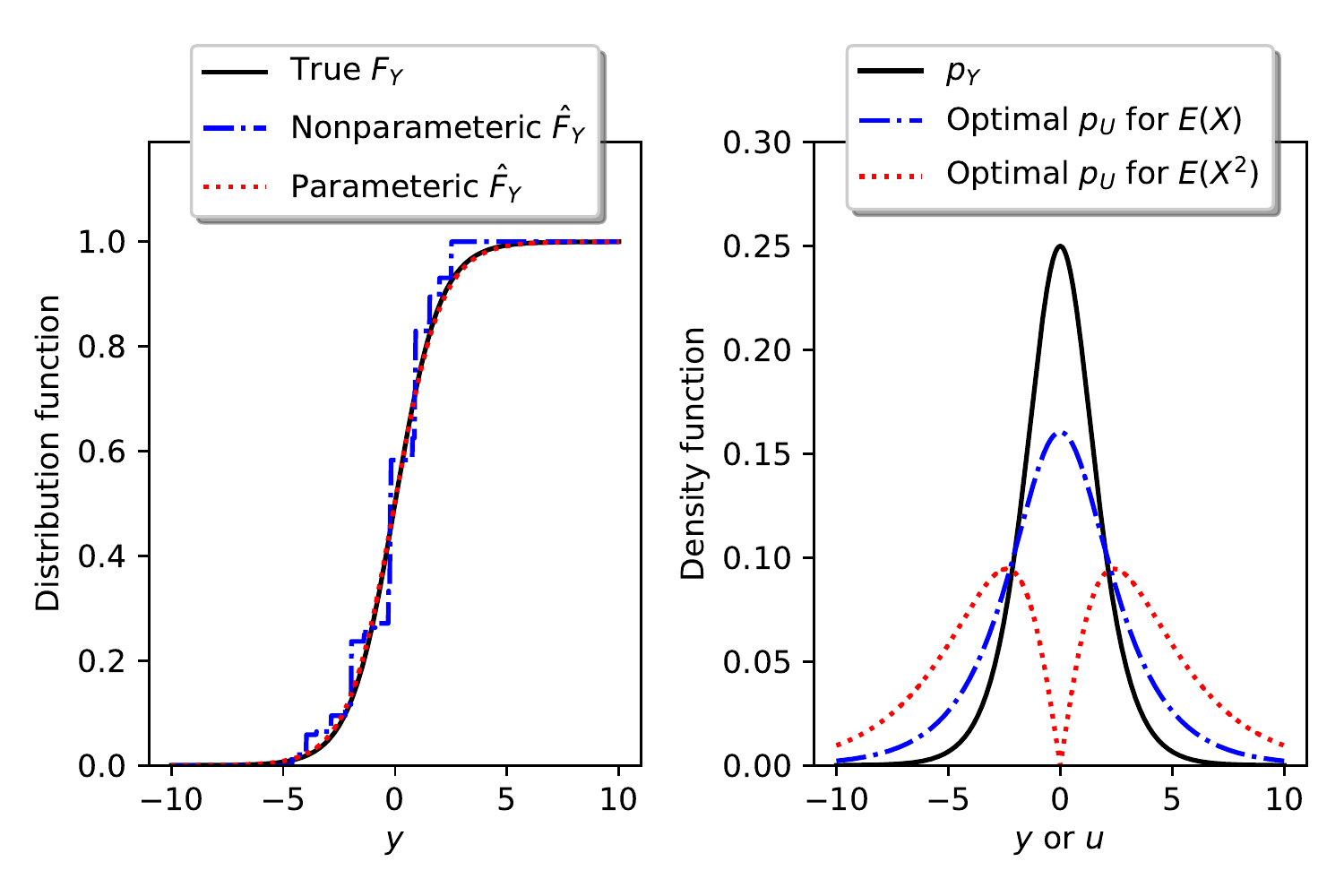}}
\vskip -0.2in
\caption{An illustration of population inference from interval-private data. The left plot shows the true CDF of the standard Logistic random variable $Y$, and its estimation using both nonparametric and parametric methods with 1000 data. The right plot shows the density of $Y$ and optimal densities of its Case-I anchor ($U$) for estimating the first and second moments (using Theorem~\ref{thm_optimalU}).}
\label{fig_cdf}
\end{center}
\vskip -0.1in
\end{figure}

\indent 
\textbf{5) Privacy guarantee}:
In the above discussion of achievability and additional properties to be introduced in Subsection~\ref{subsec_canonical_properties}, we use the privacy coverage $\E(L(S_{Z}))$ to quantify the privacy of a mechanism. 
A skeptical reader may ask how to ensure individual-level privacy. Recall that in Subsection~\ref{subsec_IP_def}, we introduced the individual privacy coverage $L(S_z)$, namely the size of an interval represented by $z$, to quantify an individual's privacy.
We provide two general ways to enhance individual-level privacy. Suppose that an individual has an associated `bottom line' $\tau \in [0,1]$, a value such that an organization can only collect an $z$ if $L(S_z) \geq \tau$. The first method uses an interval mechanism where each generated interval has coverage of at least $\tau$. Though simple, such a mechanism may not exist for some $\tau$ (e.g., $\tau=0.6$) since we cannot have two intervals whose sizes are both at least $0.6$. Moreover, the simple Case-I\&II mechanisms cannot simultaneously guarantee individual-level privacy and distributional identifiability, and thus a \textit{more general topology} (of data ranges) is required in the mechanism design. More on this will be discussed in Subsection~\ref{subsec_topology}.

The second method simply lets an individual decide whether to report the associated interval or not, depending on the $\tau$. In practice, this can be implemented in a way illustrated in Fig.~\ref{fig_example_age}(b), which provides a `Not wish to answer' option. Meanwhile, the interval mechanism needs to randomly subsample reported intervals to avoid the inference of the unreported interval (especially for a large $\tau$). Although such a mechanism introduces a selective bias (towards large intervals), we will show that the above appealing properties (such as non-informativeness and distributional identifiability) can still hold. More technical discussions are in Subsection~\ref{subsec_individual_privacy}.
 
\indent 
\textbf{6) Adaptivity to individual-level privacy}:
The interval mechanism can be extended to a progressive version. We illustrate this point in Fig.~\ref{fig_example_age}(c). Suppose for the first question, an individual chooses $\leq \Q^{(1)}$; our interface then generates another question with $\Q^{(2)} < \Q^{(1)}$; if the individual chooses $> \Q^{(2)}$, the interval $(\Q^{(2)}, \Q^{(1)}]$ is then collected. Such a progressive mechanism aligns with the above discussion of point (5), where the idea is to respect each individual's privacy while exploiting heterogeneous privacy sensitivities. We will revisit this idea in Subsections~\ref{subsec_individual_privacy} and~\ref{subsec_data1}.

\section{General Interval Mechanisms, Data Formats, and Theoretical Foundations}\label{sec_general_foundation}

With the high-level explanation in Subsection~\ref{subsec_example}, we now introduce general interval mechanisms and technical details. 

\subsection{Canonical Interval Mechanism}\label{subsec_canonical_mechanism}

Recall that a benign data collector is only interested in the population in local privacy settings instead of individual-level information.
Our interval privacy mechanism does not collect $Y$ itself but a privatized data $Z$ motivated by the scenarios where an individual will

\noindent $\bullet$ report an interval that contains $Y$,

\noindent $\bullet$ report $Y$ if it falls into an `acceptable' range, and

\noindent $\bullet$ have an acceptable range independent of $Y$.

\begin{figure}[tb]
\begin{center}
\centerline{\includegraphics[width=0.8\columnwidth]{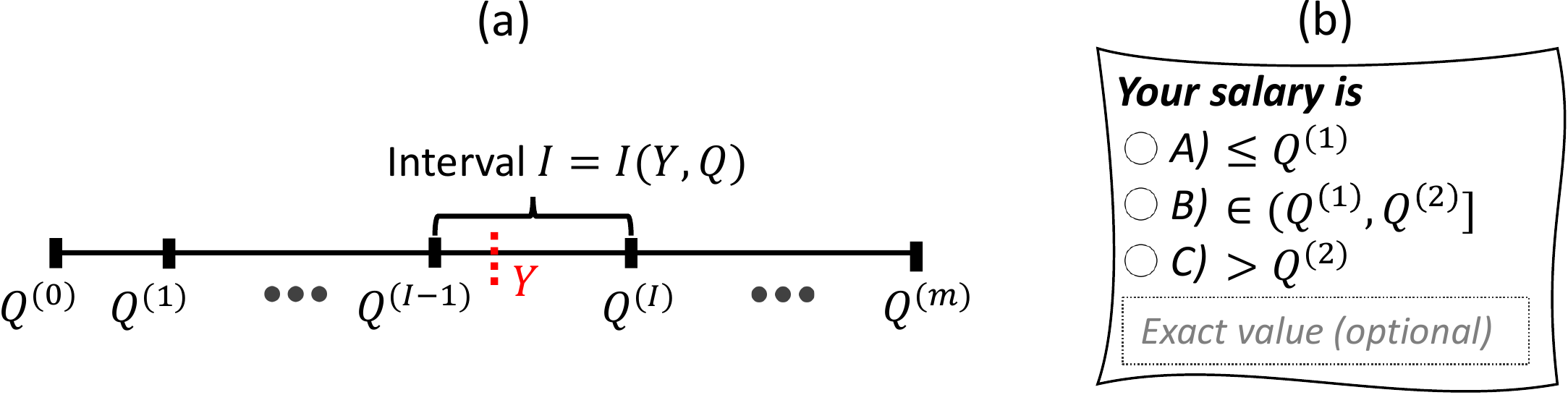}}
\vskip -0.1in
\caption{An illustration of (a) the canonical interval mechanisms, and (b) a survey-based practical interface.}
\label{fig_canonical}
\end{center}
\vskip -0.2in
\end{figure}

A natural mechanism to realize interval privacy is randomly partitioning the data domain $\Y \subset \R$ into disjoint intervals for each data owner and collecting the interval into which $Y$ falls. 
As such, we can naturally implement the mechanism through multi-choice survey questions, where each interval corresponds to a choice. This is illustrated in Fig.~\ref{fig_canonical}(a)(b). Unlike existing survey systems, our proposed system generates random (and thus different) choices for respondents. The randomness is needed to nonparametrically reconstruct the unknown population distribution from collected data, which will be elaborated in Subsection~\ref{subsec_topology}. Occasionally, the underlying point $Y$ is reported if it falls into a range that the data owner considers non-sensitive. This can be implemented by an optional text box in the above survey, as shown in Fig.~\ref{fig_canonical}(b). Consequently, the collected data are in the form of intervals or a mixture of intervals and points. 

Formally, we introduce the following notions. Let $\Q=[\Q^{(1)},\ldots,\Q^{(m-1)}]$ be a random vector with $\a=\Q^{(0)}=\Q^{(1)}<\cdots<\Q^{(m)} =\b$, as illustrated in Fig.~\ref{fig_canonical}(a). We will refer to each $\Q^{(i)}$ as an \textit{anchor} point, and let $\Ra^{(i)}\de (\Q^{(i-1)}, \Q^{(i)}]$ for $i\in [1:m]$. Then, the interval $\Ra^{(i)}$ into which $Y$ falls is collected. 
Suppose that when $Y$ falls into a pre-determined set $\C \subseteq \R$, named an \textit{acceptable range}, then the data owner chooses to disclose the value of $Y$.  
In practice, the acceptable range is at the data owner's discretion, and the set may not be fixed. 
To model the real-world complexity, we suppose that $\C$ can be one of the following: \textit{$\emptyset$, a fixed set, or the union of $(\Q^{(k-1)}, \Q^{(k)}]$ for a fixed set of $k$}. We suppose that the form of $\C$ is pre-specified and independent of $Y$.

\begin{definition}[Canonical Interval Mechanism]\label{def_mechanism}
	A privacy mechanism, denoted by $\M: Y \mapsto Z$, maps $Y$ to 
	\begin{align}
		Z = [ \Q, \, I(\Q,Y), \, Y \cdot \i_{Y \in \C}], \label{eq_format}
	\end{align}
	where $\Q=[\Q^{(1)},\ldots,\Q^{(m-1)}] \in \mathbb{R}^{m-1}$ is a random vector independent with $Y$, and $I: (\Q,Y) \mapsto i$ is the indicator function defined by $Y$ falling into $(\Q^{(i-1)}, \Q^{(i)}]$, $i\in [1:m]$.
	The corresponding privacy coverage and privacy leakage follow Definition~\ref{def_IP}.
\end{definition}

\begin{theorem}[Validity] \label{thm_validity}
  A  mechanism $\M$ in Definition~\ref{def_mechanism} satisfies the interval privacy in Definition~\ref{def_IP}. 
\end{theorem}

\begin{remark}[Interpretation of Theorem~\ref{thm_validity}]\label{remark_validity}
Intuitively, the validity is because $I(\Q, Y)$, the informative part of $Z$, only reveals the range information regarding $Y$ but not any distributional information within that range. It also allows the range to degenerate to a point when $\i_{Y \in \C}=1$ (if $\C$ is not empty). Thus, the posterior density ratio equals the prior density ratio up to a range, as shown in (\ref{eq_def}). 
Also, we point out that interval mechanisms are adaptive to an individual user's privacy preference, meaning that progressively refined information can be obtained at the discretion of individual respondents without violating Definition~\ref{def_IP}. Formally, suppose that the system also generates a second mechanism $\M': Y \mapsto Z'$ based on anchor points $\Q'$ that are (adaptively) supported on the inferred range from a previous mechanism $\M$. Then, the joint of these two mechanisms is a mechanism that meets interval privacy. This observation can be proved similarly to Theorem~\ref{thm_validity}. A practical implication is that respondents may choose to answer zero, one, or more times of randomly generated surveys depending on their earlier answers and privacy preference. This point will be revisited in Example~\ref{eg_progression}.
\end{remark}

\begin{remark}[Practical Implementation]
In practice, a privacy-preserving data collection system involves two parties, a data owner (`Alice') and a data collector (`Bob'). A general collection procedure is outlined as follows. 
First, a system designer, who may or may not be one of the two parties, define a way of generating $\Q$. Such a generating process may be open-source implemented so that it is transparent to both parties. 
Second, the two parties agree on using the mechanism for data collection. 
Third, for Alice's data value $Y$, an instance of $\Q$ is generated, and Alice reports the interval to Bob. Additionally, Alice has the option to report the exact value, but this is {at Alice's discretion}.
In the end, the set of data Bob collects consists of intervals and possibly some exact values (degenerate intervals). 
\end{remark}  

\begin{remark}[Interpretation of Data]
The observables include a partition of $\Y$ (by $\Q$), the interval that $Y$ falls (by $I(\Q, Y)$), and sometimes the value of $Y$ (represented by $Y  \i_{Y \in \C}$).
The information obtained from the privatized data $Z$ is an interval containing $y$. 
The interval-private data do not contradict the underlying truth.  
This property does not hold for popular approaches where perturbations are injected into the raw data. 

An alternative notation to $I(\Q,Y)$ is to use $m$ indicator variables $\i_{Y\leq \Q^{(1)}}, \ldots, \i_{Y\leq \Q^{(m)}}$ to represent where $Y$ is located at. By the definition, $\i_{Y\leq \Q^{(i)}}=1$ if $i \leq I(\Q,Y)$ and $\i_{Y\leq \Q^{(i)}}=0$ otherwise. 
The values of $\Q$ are random so that it is possible to identify the population distribution of $Y$ (elaborated in Subsection~\ref{subsec_topology}). So then, the randomness of $Z$ conditional on $Y$ comes from $\Q$.
The choice of $\Q$ determines privacy-utility tradeoffs. Consider an extreme case where $m$ is sufficiently large. Then, the collected interval tends to be narrow, and the privacy coverage tends to zero. In another case where $m=1$ and $\Q \in \R$ has a considerable variance, the interval is likely to be close to $(-\infty,\infty)$, which enjoys good privacy but offers little utility in distribution estimation. 
\end{remark}

\begin{remark}[Interpretation of $\C$] \label{remark_A}
The acceptable range $\C$ is a mathematical abstraction of the possibility that Alice optionally reports the raw data.
In Definition~\ref{def_mechanism}, an empty set $\C$ corresponds to the case where all observables are intervals. 
To interpret, a random set $\C$ means individuals' acceptable ranges vary (e.g., due to natural randomness), while
a deterministic $\C$ means a fixed acceptable range uniformly for all individuals.  
From Subsection~\ref{subsec_topology} and afterward, we will elaborate on the $\C=\emptyset$ case and show that the population distribution is identifiable even without exact values of $Y$.

An alternative definition of privacy leakage is $L(\C)$, meaning the probability of observing the exact value of $X$. 
Compared with the recommended $1-\tau(\M)$, the leakage here does not consider the intervals outside $\C$. For example, in the particular case $\C=\emptyset$, we have $L(\C)=0$, which is not appealing as the quantization also provides information. 

\end{remark}

\subsection{Fundamental Properties of Interval Mechanism} \label{subsec_canonical_properties}
 
In this section, we show some desirable properties of canonical interval mechanisms. They can be directly extended to other mechanisms in later sections. 

\noindent\textbf{Composition}.
Suppose there are $k$ interval-private algorithms (or collectors), each querying the same data with a mechanism $\M_j: Y \mapsto Z_j$, $j\in [1:k]$. They may collaborate to narrow down the interval that contains a particular $Y$.
This motivates the following ensemble mechanism, denoted by $\oplus_{j=1}^k \M_j$, which is an interval mechanism induced by the intersections of anchor points and the union of acceptable ranges. 
\begin{definition}[Ensemble Mechanism]\label{def_ensemble}
	The ensemble of two privacy mechanisms 
	$\M_i: X \mapsto Z = [\Q_{[j]}, \, I(\Q_{[j]},Y), \, X \cdot 1_{X \in \C_{[j]}}]
	$
	with $j=1,2$ is defined by $\M_1 \oplus \M_2:$
	\begin{align}
		X \mapsto Z = \{ \Q_{[1]}\oplus \Q_{[2]}, \, I(\Q_{[1]}\oplus \Q_{[2]},Y), \, X \cdot 1_{X \in \C_{[1]} \cup \C_{[2]}}\}, \nonumber
	\end{align}
	where $\Q_{[1]}\oplus \Q_{[2]}$ denotes the vector of all the anchor points from $\Q_{[1]}$ and $\Q_{[2]}$, and $\C_{[1]} \cup \C_{[2]}$ denotes the union of two sets $\C_{[1]},\C_{[2]}$.
	In general, the ensemble of $k$ privacy mechanisms, denoted by $\oplus_{i=1}^k \M_i$, is recursively defined by $\oplus_{i=1}^k \M_i = (\M_1 \oplus \cdots \oplus \M_{k-1}) \oplus \M_k$ ($k \geq 2$).
\end{definition}

\begin{theorem}[Composition Property]\label{thm_ensemble}  
  	Let $\M_1,\ldots,\M_k$ be $k$ interval mechanisms as in Definition~\ref{def_mechanism}.
	Then, we have
	$
		1-\tau\bigl(\oplus_{j=1}^k \M_j\bigr) 
		\leq \sum_{j=1}^k (1-\tau(\M_j)).
	$
\end{theorem}
An interpretation of the above theorem is that the privacy leakage of any ensemble mechanism is no larger than the sum of each of them. 
It is worth noting that $\Q^{[j]}$'s may or may not be independent of each other, so communications between observers are allowed for this composition property to hold. 
In other words, this composition property holds even if the $k$ mechanisms are adaptively chosen.

\noindent\textbf{Preprocessing}.
Suppose that $g: Y \mapsto g(Y)$ is a measurable function on $\Y$. 
Let $\C_g = \{g(y): y \in \C\}$ be the acceptable range for $g(Y)$, which is carried over from $\C$.
Suppose that an interval mechanism is applied to $g(Y)$ instead of $Y$ itself, with
$
	\M: g(Y) \mapsto Z = [ \Q, \, I(\Q, g(Y)), \, Y \cdot \i_{g(Y) \in \C_g}].
$
This corresponds to the `pullback' privacy mechanism
\begin{align}
	\M_g: Y \mapsto Z_g = \{ g^{-1}(\Q), \, I(g^{-1}(\Q), Y), \, Y \cdot \i_{Y \in \C}\}.	\nonumber
\end{align}
Here, $g^{-1}(\Q)$ denotes the partition of $\Y$ induced by the partition on $g(\Y)$ using $\Q$.

\begin{theorem}[Robustness to Preprocessing]
\label{thm_transform}
	For any interval mechanism $\M$, it holds that $\tau(\M_g) \geq \tau(\M)$, where the equality holds if and only if $L(g^{-1}(\C_g)) = L(\C)$.
\end{theorem}

The above result shows that if a $\tau$-interval private observation is made on a transformation of $Y$, namely $g(Y)$, the privacy coverage of the raw data $Y$ is not smaller than $\tau$, or equivalently, the leakage at the raw data domain is no larger than $1-\tau$.
Furthermore, the $\M_g$ here may be regarded as $\M_j$ in Theorem~\ref{thm_ensemble}, so Theorem~\ref{thm_ensemble} also holds for $k$ observers that may target transformations of $Y$ instead of $Y$ itself.
The inequality in Theorem~\ref{thm_transform} is strict when, e.g., $Y$ is standard Gaussian, $g(y)=y^2$, and $\C=(-\infty,0]$.

\vspace{0.1cm}
\noindent\textbf{Postprocessing}.
The next result shows that the privacy leakage is not increased by subsequent processing of $Z$. 
\begin{theorem}[Robustness to Post-processing] \label{thm_post}
	Suppose that $\M: Y \mapsto Z$ is an interval mechanism with $\tau$-interval privacy.
	Let $f: Z \mapsto W$ be an arbitrary deterministic or random mapping that defines a conditional distribution $W \mid Z$. Then $f \circ \M: Y \mapsto [Z, W]$ also meets $\tau$-interval privacy.
\end{theorem}

The above result is conceivable because $Y \rightarrow Z \rightarrow W$ is a Markov chain, and thus adding $W$ does not reveal more about the range of $Y$. We use $[Z, W]$ instead of $W$ in defining $f \circ \M$ because it is a complete observation.
The amalgamation of composition property and robustness permits modular designs and analyses of interval mechanisms.

\subsection{Extension: Interval Mechanism of General Topology} \label{subsec_topology}

It is natural to extend the canonical interval mechanism in Subsection~\ref{subsec_canonical_mechanism} by considering a partition of $\R$ into $m$ general ranges, denoted by $\{\Ra^{(i)}\}_{i=1}^m$. We suppose each $\Ra^{(i)}$ is a Borel set to define probability on them properly. Also, to operate data collection in practice, we let such a partition be determined by a fixed-dimension random vector $\T \in \R^q$, and both $m,q$ be fixed positive integers. Formally, we introduce the following notion. We omit the acceptable range from now on for notational simplicity.

\begin{definition}[Extended Interval Mechanism]\label{def_extended_mechanism}
	Suppose that $\T $ is a $q$-dimensional random vector independent with $Y$. Let 
	\begin{align}		
	\Ra: t \mapsto \{\Ra_t^{(i)}\}_{i=1}^m \label{eq_Ra}
	\end{align}
	denote a map from each $t \in \R^q$ to a partition of $\R$. 
	Let $I: (\Ra,t,y) \mapsto i$ denote the indicator function defined by $y$ falling into $\Ra_t^{(i)}$, $i\in [1:m]$.
	A privacy mechanism, denoted by $\M: Y \mapsto Z$, maps $Y$ to  
	$
		Z \de [ \T, \, I(\Ra,\T,Y)] . 
	$
\end{definition}

A particular case is when $q=m-1$, $\T^{(i)}=\Q^{(i)}$, and $\Ra_{\T}^{(i)}=(\Q^{(i-1)}, \Q^{(i)}]$ for $i\in [1:q]$, which corresponds to Definition~\ref{def_mechanism}.
It can be verified that the extended mechanism still satisfies the interval privacy in Definition~\ref{def_IP} and all the properties in Subsection~\ref{subsec_canonical_properties}.
Next, we first explain why such extended mechanisms can be practically interesting. We then introduce the estimation of $\Fy$ and sufficient conditions to guarantee the distributional identifiability.

Consider the setting where we want to ensure a lower bound on each individual's privacy coverage, namely $L(S_z) \geq \tau$ for each collected $z$ for a given $\tau>0$. A canonical mechanism will violate the distributional identifiability. To see that, let us consider the left-most interval $(\Q^{(0)}, \Q^{(1)}]$, where $\Q^{(0)}$ is fixed (e.g., $-\infty$) and $\Q^{(1)}$ is random. To ensure  
\begin{align}
	L((\Q^{(0)}, \Q^{(1)}]) \geq \tau, \label{eq_low}
\end{align}
the smallest anchor point $\Q^{(1)}$ cannot take values in $(-\infty, \Fy^{-1}(\tau)]$. Then, the distribution of $Y$ is not identifiable on the left tail.  
To address the issue, we consider the following example of Definition~\ref{def_extended_mechanism} that is not a canonical mechanism.

\begin{example} \label{eg_ring_topology}
Without loss of generality, suppose that $\T^{(1)} < \cdots < \T^{(q)}$ almost surely. Let $\Ra_{\T}^{(1)}=(-\infty, \T^{(1)}] \cup (\T^{(q)}, \infty)$, and $\Ra_{\T}^{(i)}=(\T^{(i-1)}, \T^{(i)}]$ for $i=2,\ldots,q$. In other words, we concatenate the left-most and right-most canonical intervals into one range. This is naturally represented by a `ring' topology as illustrated in Fig.~\ref{fig_topology_example}(a), where $\pm \infty$ collapse into one anchor. In that figure, we used the notation of $\Q^{(i)}$ (instead of $\T^{(i)}$) for an easier comparison with Fig.~\ref{fig_canonical}(a). 
We also show a simple interface example in Fig.~\ref{fig_topology_example}(b), which is the counterpart of Fig.~\ref{fig_canonical}(b). In this example, if $\Q$ is designed so that $\P_Y(Y \in (\Q^{(1)}, Q^{(2)}]) \in [\tau, 1-\tau]$ (which implicitly requires $\tau \leq 0.5$), it is intuitively possible to identify $\Fy$. Next, we provide a formal method to guarantee the distributional identifiability. Our result is a nontrivial generalization of the existing theory for the Cases I\&II interval data~\cite{groeneboom1991nonparametric}.
\end{example}

\begin{figure}[tb]
\begin{center}
\centerline{\includegraphics[width=0.7\columnwidth]{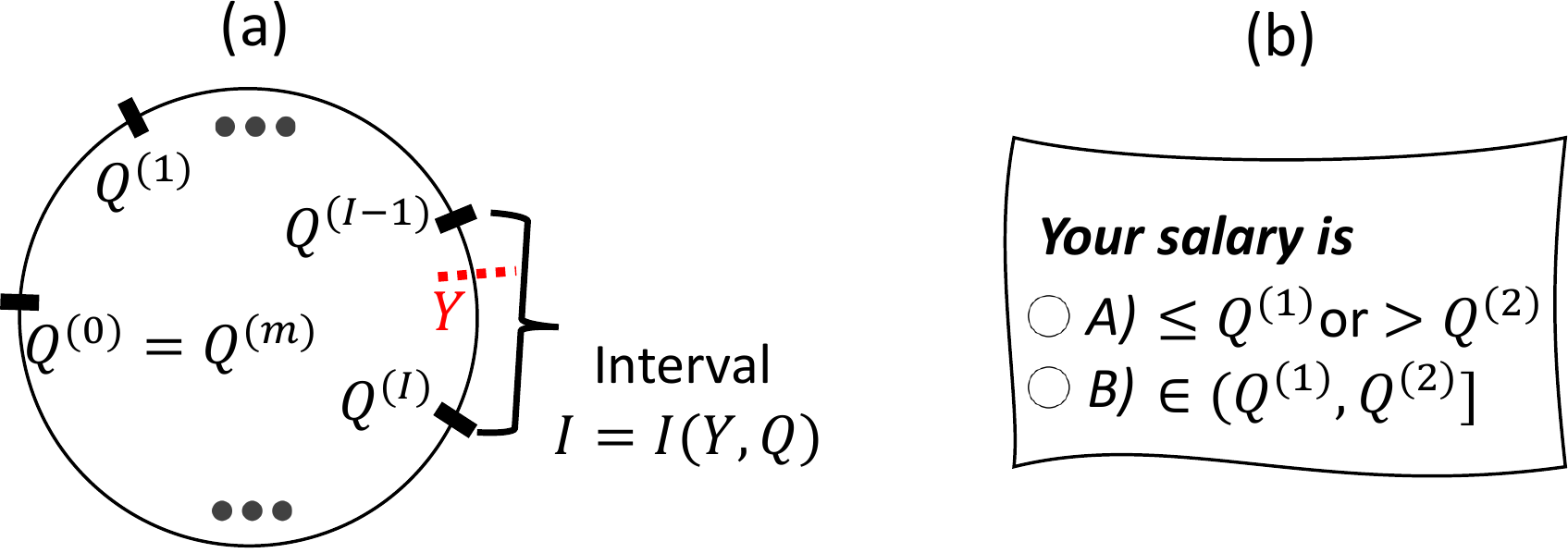}}
\vskip -0.1in
\caption{An extended interval mechanism and its example interface.}
\label{fig_topology_example}
\end{center}
\vskip -0.1in
\end{figure}

\noindent\textbf{Nonparametric maximum likelihood estimator (NPMLE)}:
We let $\B^{(i)}(\T, Y)\de \i_{Y \in \Ra_{\T}^{(i)}}$, or $\B^{(i)}$ for brevity, for $i\in [1:m]$.
Let $Z_j = [ \T_j, \, I(\Ra,\T_j,Y_j)]$, or equivalently, $Z_j = [ \T_j, \, \B_j^{(1)}, \ldots, \B_j^{(m)}]$, $j\in [1:n]$ denote the observed data, where $n$ is the sample size.
For any right-continuous distribution function $F$ on $\Y$, let $F(r)$ denotes the corresponding probability of a Borel set $r$.
To estimate the underlying distribution of $Y$ without parametric assumptions\footnote{If $\Fy$ is parameterized by a fixed-dimensional parameter, standard maximum likelihood estimation and asymptotics can be readily applied~\cite[Ch.5]{van2000asymptotic}.}, we consider the log-likelihood functional 
\begin{align}
	\psi: F \mapsto 
	&\int_{\R^q \times \Y} \sum_{i=1}^m \B^{(i)} \log F(\Ra^{(i)}) \, d \P_n(t,y) \nonumber \\
	&\de\frac{1}{n}\sum_{j=1}^n \sum_{i=1}^m \B_j^{(i)} \log F(\Ra_j^{(i)}) , \label{eq_loglik}
\end{align}
where $\P_n(\cdot,\cdot)$ denotes the empirical probability measure from $[\T_j,Y_j]$, $j\in [1:n]$.
We define NPMLE as a right-continuous distribution function $\hat{F}_n$ that maximizes $\psi(F)$. An example of the related quantities are visualized in Fig.~\ref{fig_NPMLE}. Note that the objective in (\ref{eq_loglik}) can only be defined up to the values of $\hat{F}_n$ at the anchor points from the intersections of observed ranges (e.g., $y_1,y_2,y_3$ in Fig.~\ref{fig_NPMLE}). As such, we consider the NPMLE as a piecewise function with only jumps at anchor points. 
Next, we provide conditions for interval mechanisms to preserve distributional information, namely distributional identifiability.

\begin{figure}[tb]
\begin{center}
\centerline{\includegraphics[width=0.5\columnwidth]{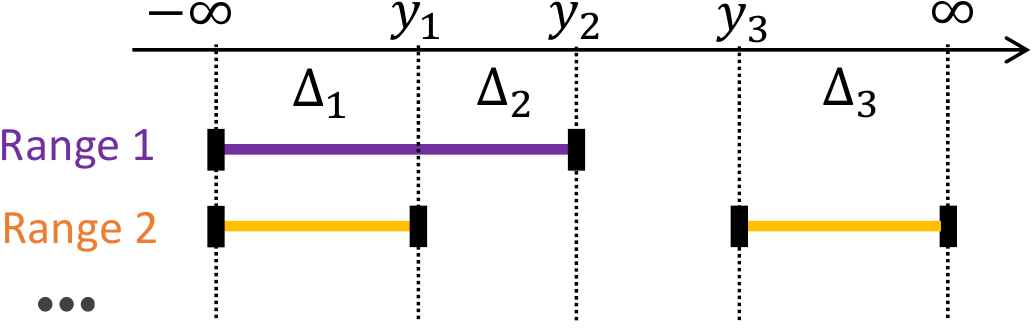}}
\vskip -0.1in
\caption{An example of the collected range data, in which case the log-likelihood function is $F \mapsto \log F(y_2) + \log (F(y_1)+1-F(y_3)) + \cdots$, or represented by $[\Delta_1,\Delta_2,\Delta_3] \mapsto \log (\Delta_1+\Delta_2) + \log (\Delta_1+\Delta_3) + \cdots$.}
\label{fig_NPMLE}
\end{center}
\vskip -0.1in
\end{figure}

\vspace{0.1cm}
\noindent \textbf{Resolvability condition}: $T$ has a density with respect to the Lebesgue measure. For each $y$ in the closure of $\Y$, there is an open neighborhood $N(y)$ such that for all $y_a,y_b\in N(y)$, the interval $(y_a,y_b]$ satisfies: there exist $t, t'$ in the essential support of $T$ and ranges $\ra^{(k)} \in \Ra_t, \ra^{(\ell)} \in \Ra_{t'}$ ($1 \leq k, \ell \leq m$) such that $(y_a,y_b] \cap \ra^{(k)}=\emptyset$ and $(y_a,y_b] \cup \ra^{(k)}=\ra^{(\ell)}$.

The resolvability is determined by both the support of $T$ and topology $\Ra$ introduced in (\ref{eq_Ra}). Intuitively speaking, a mechanism is resolvable if any small interval of $\Y$ can be the difference between two feasible ranges. It will be used in our proof in the following way. We will first prove that $\hat{F}_n(\ra) \approx \Fy(\ra)$ for each feasible range $\ra$. This, together with resolvability, gives $\hat{F}_n \approx \Fy$ on any small interval $(y_a,y_b]$, which further implies $\hat{F}_n \approx \Fy$ globally. 
For example, it can be verified that any canonical mechanism (Definition~\ref{def_mechanism}) is resolvable if its $\Q$ has a positive density wherever $0<\Fy(\Q^{(1)})< \cdots <\Fy(\Q^{(m-1)})<1$. On the other hand, a canonical mechanism satisfying (\ref{eq_low}) is not resolvable, shown in Fig.~\ref{fig_topology_privacy}(a). Also, a ring design in Example~\ref{eg_ring_topology} is resolvable if $\Q^{(1)}$ has a positive density on $\R$ and $\Q^{(2)}$ conditional on $\Q^{(1)}$ does not degenerate to a point, shown in Fig.~\ref{fig_topology_privacy}(b).

We will also need the following condition. Let $G \circ F$ denote the function that maps $y$ to $G(F(y))$. Recall the $R,T,q$ in Definition~\ref{def_extended_mechanism}. As before, with a slight abuse of notation, we use $F(r)$ and $F(y)$ to denote the probability of a Borel set $r$ and the CDF at a point $y$, respectively.

\vspace{0.1cm}
\noindent \textbf{Monotonicity condition}: For any CDF $F$, the map $R$ satisfies: a) for each $j \in [1:q]$ and $i \in [1:m]$, $F(\Ra_t^{(i)})$ is either non-decreasing or non-increasing in $t^{(j)}$ with $t^{(j')}$ ($j'\neq j$) fixed; b) there exist functions $G^{(1)},\ldots,G^{(q)}$ that are nondecreasing, continuous, and bounded on $[0,1]$ such that for each $j \in [1:q]$, $i \in [1:m]$, and $t_a, t_b \in \R^q$ differing only in the $j$-th entry, we have $|F(\Ra_{t_b}^{(i)}) - F(\Ra_{t_a}^{(i)})| \leq |G^{(j)} \circ F(t_b^{(j)}) - G^{(j)} \circ F(t_a^{(j)})|$.

Intuitively speaking, condition (a) states that the probability mass on each range is monotone in each entry of $t$, and (b) means that the sensitivity of those probabilities can be controlled in terms of $t$. It can be verified that the above condition holds for all canonical mechanisms and the example in Fig.~\ref{fig_topology_example}, with $G$ being the identity map.

\begin{theorem} \label{thm_NPMLE}
	Assume that an extended interval mechanism satisfies the above Resolvability and Monotonicity conditions, and $\Fy$ is continuous.
	Then, $\sup_{y\in\Y}|\hat{F}_n(y) - \Fy(y)| \rightarrow 0$ almost surely as $n \rightarrow \infty$.
\end{theorem}

\begin{figure}[tb]
\begin{center}
\centerline{\includegraphics[width=0.8\columnwidth]{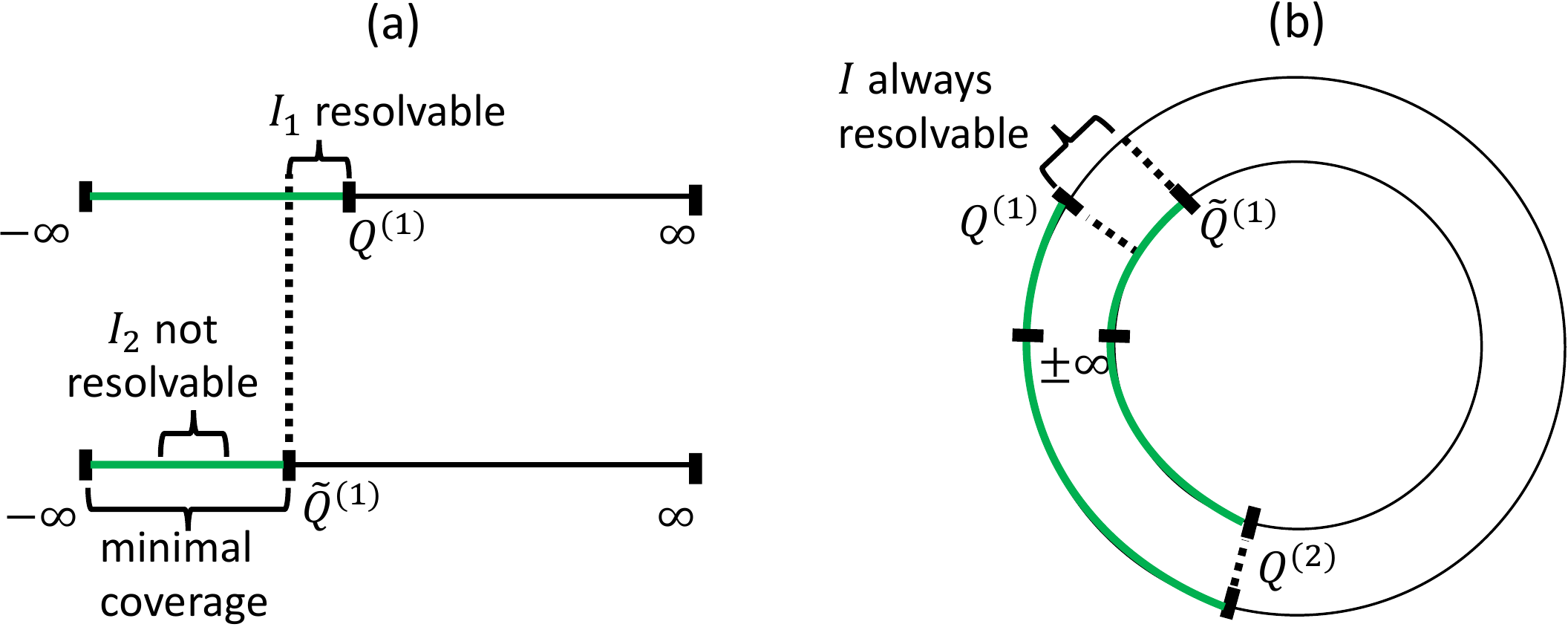}}
\vskip -0.1in
\caption{Interval mechanisms of different topologies, in which resolvability is (a) not met, and (b) met.}
\label{fig_topology_privacy}
\end{center}
\vskip -0.1in
\end{figure}

\subsection{Extension: Individual-Level Privacy Enhancement} \label{subsec_individual_privacy}

In Subsection~\ref{subsec_topology}, we considered the problem to ensure a lower bound on each individual's privacy coverage. Specifically, for each individual $z$ collected, we want to ensure that 
\begin{align}
L(S_z) \geq \tau, \label{eq_210}
\end{align}
where $S_z$ denotes the corresponding range.
Note that this requirement is for each individual, much stronger than lower-bounding population coverage. 
As shown in Example~\ref{eg_ring_topology}, a general approach is to design an extended interval mechanism such that the coverage of each feasible range is lower bounded, namely
$
	L(\Ra_T^{(i)}) \geq \tau \textrm{ almost surely, } \, \forall i \in [1:m] .
$
A limitation of this approach is that it requires $\tau \leq 1/m$.
What if the privacy system or an individual requires a large $\tau$, say $0.6$?
We propose an alternative approach below.

The key idea of the alternative approach is only to collect ranges that satisfy (\ref{eq_210}). Depending on practical needs, it can be implemented and interpreted in two ways. \\
\indent \textit{Way I}: We allow an individual to choose `Not wish to answer' as shown in Fig.~\ref{fig_example_age}(b), so $\tau$ is a representation of the (possibly unknown) underlying privacy sensitivity. \\
\indent \textit{Way II}: We let the system pick up the ranges satisfying (\ref{eq_210}), where $L(\cdot)$ can be approximated using an estimate of $\Fy$, and $\tau$ represents a known system-specific privacy budget. 

Nevertheless, a skeptical individual may worry about information leakage from not collecting his/her data, especially when individuals are not de-identified (e.g., tracked by static IPs). 
This motivates the following interval mechanism. Let $\tau,\rho \in [0,1]$ denote two constants, and $\rand \sim \Bern(\rho)$ denote a Bernoulli random variable with $\P(\rand=1)=\rho$.

\begin{definition}[Selective Mechanism]\label{def_selective_mechanism}
	With $Z$ in Definitions~\ref{def_IP} or~\ref{def_extended_mechanism}, a $(\tau,\rho)$-selective mechanism collects $Z$ if $L(S_Z) \geq \tau$ and $\rand =1$ simultaneously hold, and `null' otherwise, where $\rand \sim \Bern(\rho)$ is independently generated.
\end{definition}

Here, the term `null' indicates that the system does not collect the range $Y$ belongs to, and it discards all the generated ranges for that individual. 
Also, the observed data $Z$ implicitly implies $L(S_Z) \geq \tau$. 
It can be verified that a selective mechanism satisfies interval privacy.
We say that a selective mechanism has an `\textit{ignorability}' property if the probability of $\rand=0$ conditional on collecting `null' is at least $0.5$. Intuitively, not collecting data is likely due to an independently generated cutoff.

\begin{proposition} \label{prop_ignorability}
	A $(\tau,\rho)$-selective mechanism satisfies ignorability if $\rho \leq \P(L(S_Z) \geq \tau)$.
\end{proposition}

Intuitively, the smaller probability of meeting individual-level $\tau$-coverage, the smaller $\rho$ (less selection) needed for ignorability. 
Since $\P(L(S_Z) \geq \tau)$ is interpreted as the frequency of $L(S_Z) \geq \tau$, in line with the \textit{Way I} or \textit{II}, the collection system may estimate it using the response rate from historical collections, or calculate it from an estimated $\Fy$. 

Regarding the distributional identifiability, Theorem~\ref{thm_NPMLE} no longer applies because the collected ranges are selectively biased towards large coverages (at least $\tau$), causing dependence of the underlying data $Y$ and anchor $T$ (in Definition~\ref{def_extended_mechanism}). We will show that one can still guarantee the distributional identifiability for selective mechanisms under an adaption of earlier results and an additional condition. Moreover, the above discussions assume a fixed $\tau$. To accommodate individuals' heterogeneous privacy sensitivities, one may also consider a random $\tau$. Details on these are deferred to the supplement. 

Moreover, we previously considered one-time collection from each individual. In practice, when individuals have different but unknown privacy sensitivities, we consider the following extension of \textit{Way I}. If an individual chooses `Not wish to answer,' the system adaptively proceeds with a further question to the same individual conditional on the previous one. A general mechanism was explained in Remark~\ref{remark_validity}. We exemplify the idea below, also illustrated in Fig.~\ref{fig_example_age}(c). We will provide a real-data experiment in Subsection~\ref{subsec_data1}. 

\begin{example} \label{eg_progression}
	Suppose that $Y \in [\underline{U}^{(0)}, \bar{U}^{(0)}]$. Let $\mathcal{G}_{\underline{u},\bar{u}}$ denote a distribution that is determined by $\underline{u}$ and $\bar{u}$, and supported on $[\underline{u},\bar{u}]$.
	At each round $h\geq 1$, the system\\
	1) generates $U^{(h)} \sim \mathcal{G}_{\underline{U}^{(h-1)}, \bar{U}^{(h-1)}}$, $Z^{(h)}=[U^{(h)}, \i_{Y \leq U^{(h)}}]$;\\
	2) lets $\bar{U}^{(h)}=U^{(h)}$ if $Y \leq U^{(h)}$, and $\underline{U}^{(h)}=U^{(h)}$ otherwise;\\
	3) collects $\cup_{i\leq h} Z^{(h)}$ and proceeds to the next round only if $L(S_{Z^{(h)}}) \geq \tau$, and $\cup_{i\leq h} Z^{(h-1)}$ (`null' for $h=1$) otherwise.
\end{example}

\subsection{Regression with Private Responses} \label{sec_reg}

This section proposes a general approach to fit supervised regression using interval-private responses. 

Suppose that we are interested in estimating the regression function with $Y$ being the response variable and $X \in \R^p$ 
the features.
Suppose that $Y$ has been already privatized, and we only access an interval-private observation of $Y$ while the data $X$ is visible.
This scenario occurs, for example, when Alice (who holds $Y$) sends her privatized data to Bob (who holds $X$) to seek Assisted Learning~\cite{DingAssist}. The scenario also occurs when $Y$ has to be private while $X$ is already publicly available.

Following the standard setting of regression analysis, we postulate the data generating model 
$
  Y = f^*(X) + \v , 
$
where $f^*$ is the underlying regression function to be estimated, and $\v \sim \Fv$ is an additive random noise.
We suppose that $X$ is a random variable independent of $\v$, and that $\Fv$ is a known distribution, say Gaussian or Logistic distributions. 
We will discuss unknown $\Fv$ in Remark~\ref{remark_computation}.
Suppose that a Case-I interval mechanism is used and data are in the form of 
\begin{align}	
D_i=[u_i, \delta_i, x_i]^\T,  
\textrm{ where } \delta_i \de \i_{y_i \leq u_i}, \, \ i\in [1:n].	
\label{eq_D} 	
\end{align}
Since $f^*$ is unknown, a general approach is to represent $f^*$ with linear functions 
$
  Y = X^T \beta + \v \nonumber  
$,
where $\beta \in \R^p$ is treated as an unknown parameter and $\v \sim \Fv$.
The above model includes parametric regression and nonparametric regression based on series expansion (e.g., with polynomial, spline, or wavelet bases).
To estimate $\beta$ from $D_i$'s, a classical way is to maximize the likelihood, e.g., 
$\beta \mapsto \prod_{i=1}^n [ \Fv(u_i-x_i^T \beta )^{\delta_i} \{1-\Fv(u_i-x_i^T \beta)\}^{1-\delta_i}]$ for Case-I intervals.

Though the likelihood approach is principled for estimating a parametric regression, its implementation depends on the specific parametric form of the regression function $f$, and its extension to nonparametric function classes is challenging.  
In supervised learning, data analysts typically use a nonparametric approach, such as various types of tree ensembles and neural networks.
However, the existing regression techniques for point-valued responses ($Y$) cannot handle interval-valued responses. 
As such, we are motivated to `transform' the interval-data format into the classical point-data form to enable the direct use of existing regression methods and software.

Our main idea is to transform the data format from intervals to point values so that many existing regression methods can be readily applied.
We propose to use 
\begin{align}
  \tilde{Y} = \E (Y \mid D, X)   \label{eq17}
\end{align}
as a surrogate to $Y$.
This is motivated from the observation that $\tilde{Y}$ is an unbiased estimator of $f(x)$ for a given $X=x$, namely
$ \E (\tilde{Y} \mid x) = \E(Y \mid x)=f(x)$. 
Suppose that we choose a loss function $\ell: \Y \times \Y \rightarrow \R$ such as $L_1$ or $L_2$ loss. 
For notational convenience, we let $\E_n$ denote the empirical expectation, e.g., $
\E_n \ell(\tilde{Y}, f(X)) = n^{-1}\sum_{i=1}^n \ell(\tilde{Y}_i, f(X_i)).
$
Based on the above arguments in (\ref{eq17}), it is desirable to solve the following optimization problem
\begin{align}
  &\min_{f \in \F} \ \E_n \ell\bigl( \tilde{Y}, f(X) \bigr) \label{eq20} \\
  &\textrm{where } \tilde{Y} = \E (Y \mid D, X)
  = f(X) + \E(\v \mid D, X). \label{eq21} 
\end{align}

The following result justifies the validity of using $\tilde{Y}$ as  a surrogate of $Y$ to estimate $f$, if the $\tilde{Y}$ in (\ref{eq21}) were statistics. 
We define the norm of $f$ to be $\norm{f}_{\P_X} \de \sqrt{\E(f(X)^2)}$.
Suppose that the underlying regression function $f^*$ belongs to $\F$, a parametric or nonparametric function class with bounded $L_2(\P_X)$-norms.
The following result shows that the optimal $f$ obtained by minimizing (\ref{eq20}) with a squared loss $\ell$ is asymptotically close to the underlying truth $f^*$. 

\begin{theorem}[Regression Estimation] \label{prop_reg}
  Suppose that 
\begin{align}
  \sup_{f\in \F}\bigl| \E_n(\tilde{Y}-f(X))^2 - \E(\tilde{Y}-f(X))^2 \bigr| \limp 0 \nonumber
\end{align}
(convergence in probability) as $n \rightarrow \infty$. 
Then, any sequence $\hat{f}_n$ that maximizes $\E_n(\tilde{Y}-f(X))^2$ converges in probability to $f^*$ in the sense that
$
  \norm{\hat{f}_n - f^*}_{\P_X} \limp 0
$
as $n \rightarrow \infty$.
\end{theorem}

In practice, however, the calculation of $\tilde{Y}$ itself is unrealistic as it involves the knowledge of  $f(X)$.
In other words, the unknown function $f$ appears in both the optimization (\ref{eq20}) and calculation of surrogates (\ref{eq21}).
The above difficulty motivates us to propose an iterative method where we iterate the steps in (\ref{eq20}) and (\ref{eq21}), using any commonly used supervised learning method to obtain $\hat{f}$ at each step.

The pseudocode is provided in Algo.~\ref{algo1}, where Case-I intervals are considered for brevity. In practice, we set the initialization by $\hat{f}_{n,0}(x)=0$ for all $x$. We experimentally found that Algo.~\ref{algo1} is robust and works well for a variety of nonlinear models such as tree ensembles and neural networks.

  \begin{algorithm}[tb]
    \centering
    \caption{
    Interval Regression by Iterative Transformations
    (a Case-I example) 
    }\label{algo1}
    \footnotesize
    \begin{algorithmic}[1]
      \renewcommand{\algorithmicrequire}{\textbf{Input:}}
      \renewcommand{\algorithmicensure}{\textbf{Output:}}
      \REQUIRE Interval-valued responses $D_i$ in Eq.~(\ref{eq_D}) and predictors $x_i\in \R^p$, $i\in [1:n]$, function class $\F$, error distribution function $\Fv$
      \textit{Initialization}: 
      Round $k=0$, function $\hat{f}_{n,0}(\cdot)$
      \REPEAT
      \STATE Let $k \leftarrow k+1$
      \STATE Let $\tilde{u}_i=u_i - \hat{f}_{n,k-1}(x_i)$ and update the representative $\tilde{y}_i, i\in [1:n]$
      \begin{align}
      \tilde{y}_i = \hat{f}_{n,k-1}(x_i) + \E(\v \mid \tilde{u}_i,  \delta_i, x_i) .\label{eq22}
      \end{align}
      \STATE Fits a supervised model $\hat{f}_{n,k}$ using $(\tilde{y}_i, x_i)$ as labeled data by optimizing (\ref{eq20}) using a preferred method 
      \UNTIL{A stop criterion satisfied (e.g., if the fitted values do not vary much)}
      \ENSURE The estimated function $\hat{f}_{n,k}: \R^p \rightarrow \R$
    \end{algorithmic}
  \end{algorithm}

\begin{remark}[Computing the conditional expectation in (\ref{eq22})] \label{remark_computation}
The term $\E(\v \mid \tilde{u}_i,  \delta_i, x_i)$ is essentially $\E(\v \mid \v \leq u_i-\hat{f}_{k-1}(x_i))$ if $\delta_i=1$, or $\E(\v \mid \v > u_i-\hat{f}_{k-1}(x_i))$ otherwise. 
When calculating (\ref{eq22}), we need to specify a distribution for the error term $\v$. Though a misspecified distributional assumption often affects inference results~\cite{DingOverview}, we found from experimental studies that the accuracy of estimating $f$ here is not sensitive to misspecification of the noise distribution. 
A practical suggestion to data analysts is to treat $\v$ as Logistic random variables to simplify the computation. Some related experimental studies and remarks on fast computation are included in the supplement.  
Also, if the standard deviation of the noise $\v$ is unknown in practice, we suggest estimate $\sigma^2$ with $n^{-1} \sum_{i=1}^n ( y_i-\tilde{y}_i )^2$ at each iteration of Algo.~\ref{algo1}.
\end{remark}

\section{Experiments} \label{sec_exp}

We provide experiments on the use of interval privacy, including unsupervised, supervised, and real-data examples.  

\subsection{Estimation of Moments}

This experiment demonstrates moment estimation with interval-private data of reasonably broad privacy coverage. 
Suppose that $100$ data are generated from $Y \sim \mathcal{N}(0.5,1)$, and the private data are based on the Case-I mechanism with $U \in \textrm{Uniform}[-T,T]$, where $T=2n^{1/3}$. The privacy coverage is around $0.95$ (or $5\%$ leakage).
The goal is to estimate $\E(Y)$. 
We consider two methods and compare them with the baseline estimates using raw data (in hindsight) in Table~\ref{tab_moment}.  
The first method, denoted by `Example~\ref{eg2}', uses the estimator in (\ref{eq_eg2}). By a similar argument as the proof of Proposition~\ref{prop_eg2}, the choice of $T$ guarantees that the $\mu$ can be consistently estimated. 
The second method is the NPMLE implemented in the `Icens' R package~\cite{gentleman2010icens}.
We also consider two methods based on raw data: the sample average and the sample median. 
To demonstrate the robustness, we add $0\%$, $1\%$, and $5\%$ proportion of outliers (meaning $Y=999$).
We also consider the estimation of $\E(Y^2)$ in a similar setting, except that we use $U \in \textrm{Uniform}[0,2T]$ to collect $Y^2$ so that the privacy coverage remains around $0.95$.

The results summarized in Table~\ref{tab_moment} indicate that the estimation under highly private data is reasonably well when compared with the oracle approach with $0\%$ outlier. Also, the estimation from interval-private data tends to be more robust against outliers than the estimation based on the simple mean and comparable to the median (using raw data).

\begin{table}[tb]
\centering
\caption{Mean absolute errors of estimating $\E(Y)$ and $\E(Y^2)$ using interval-private data and raw data (in hindsight) subject to $0\%$, $1\%$, and $5\%$ outliers. Values are calculated from 1000 replications so that the standard errors are all within $0.01$. }
\vspace{-0.2cm}
\label{tab_moment}
\scalebox{1}{ 
\begin{tabular}{cccccccc}
\toprule
                                                                                    &        & \multicolumn{3}{c}{Estimate $\E(Y)$} & \multicolumn{3}{c}{Estimate $\E(Y^2)$} \\ \cline{2-8} 
                                                                                    & $n$    & $0\%$      & $1\%$      & $5\%$      & $0\%$       & $1\%$       & $5\%$      \\ \midrule
\multirow{2}{*}{\begin{tabular}[c]{@{}c@{}}Private data\\ (Example~\ref{eg2})\end{tabular}} & $100$  & $0.45$     & $0.44$     & $0.58$     & $0.79$      & $0.82$      & $1.03$     \\
                                                                                    & $1000$ & $0.29$     & $0.33$     & $0.99$     & $0.57$      & $0.64$      & $1.96$     \\ \midrule
\multirow{2}{*}{\begin{tabular}[c]{@{}c@{}}Private data\\ (NPMLE)\end{tabular}}     & $100$  & $0.32$     & $0.36$     & $0.92$     & $13.09$     & $11.58$     & $10.64$    \\ \cline{2-8} 
                                                                                    & $1000$ & $0.12$     & $0.21$     & $1.13$     & $3.68$      & $4.08$      & $4.45$     \\ \midrule
\multirow{2}{*}{\begin{tabular}[c]{@{}c@{}}Raw data \\ (Mean)\end{tabular}}         & $100$  & $0.08$     & $9.98$     & $49.93$    & $0.14$      & $9.98$      & $49.89$    \\ \cline{2-8} 
                                                                                    & $1000$ & $0.03$     & $9.98$     & $49.92$    & $0.04$      & $9.98$      & $49.88$    \\ \midrule
\multirow{2}{*}{\begin{tabular}[c]{@{}c@{}}Raw data \\ (Median)\end{tabular}}       & $100$  & $0.11$     & $0.10$     & $0.12$     & $0.66$      & $0.65$      & $0.60$     \\
                                                                                    & $1000$ & $0.03$     & $0.03$     & $0.07$     & $0.67$      & $0.65$      & $0.60$     \\ \bottomrule
\end{tabular}
}
\end{table}

\subsection{Estimation of Regression Functions}\label{subsec_reg_exp}

We first demonstrate the method proposed in Subsection~\ref{sec_reg} on the linear regression model $Y = f(X) + \v$, where $f(X)=\beta X$ is to be estimated from the Case-I privatized data $Z=[U,\Delta]$. 
We generate $n=200$ data with $\beta=1$, $X,\v\sim_{i.i.d.} \mathcal{N}(0,1)$, $U$ a Logistic random variables with scale~$5$. The corresponding privacy coverage is around $0.9$.  
Fig.~\ref{fig_V4_linearReg} demonstrate a typical result. The prediction error is evaluated by mean squared errors $\E(f(\tilde{X})-\hat{f}(\tilde{X}))^2$ where $\tilde{X}$ denotes the unobserved (future) data. 
With a limited size of data, the algorithm will produce an estimate $\hat{f}(x)=\hat{\beta}X$ that converges well within 20 iterations.
The initialization is done by simply setting $\hat{f}_{n,0}(x)=0$.

\begin{figure}[tb]
\begin{center}
\centerline{\includegraphics[width=0.9\columnwidth]{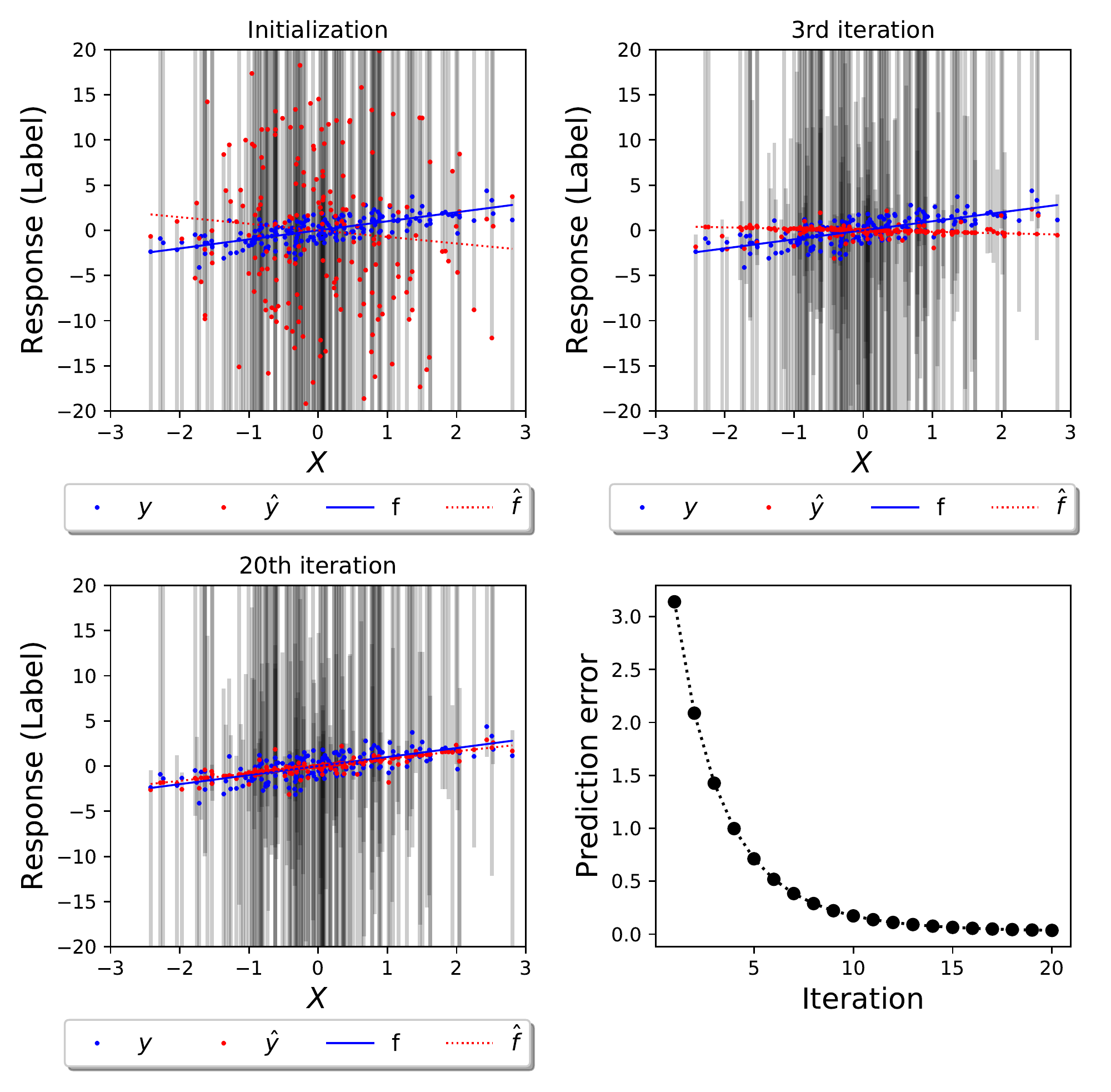}}
\vskip -0.1in
\caption{Experiments in Subsection~\ref{subsec_reg_exp}: Snapshots of Algo.~\ref{algo1} for linear regression at the 1st (left-top), 3rd (right-top), and 20th (left-bottom) iterations, and the prediction error ($L_2$ loss) versus  iteration (right-bottom). Grey vertical segments indicate the observed intervals in the form of $(-\infty,u]$ or $(u,\infty)$; Blue dots and lines indicate the unprotected data $Y$ and the underlying true regression function; Red dots and dashed lines indicate the adjusted data $\tilde{Y}$ in (\ref{eq17}) and the estimated regression function. }
\label{fig_V4_linearReg}
\end{center}
\vskip -0.2in
\end{figure}
 
In another experiment, we demonstrate the method proposed in Subsection~\ref{sec_reg} on the nonparametric regression model 
$Y = f(X) + \v$, using the Case-II privatized data $Z=[U,V,\Delta,\Gamma]$.
Suppose that $n=200$ data are generated from quadratic regression $f(X)=X^2 - 2 X + 3$, and  $\v \sim \mathcal{N}(0,1)$. Let $U=\min(L_1,L_2)$, $V=\max(L_1,L_2)$, where $L_1,L_2$ are independent Logistic random variables with scale~$5$. The corresponding privacy coverage is around $0.9$. 
Random Forest (depth 3, 100 trees) with features $X_1=X$, $X_2=X^2$ are used to fit Algo.~\ref{algo1}.
Fig.~\ref{fig_V4_quadReg} demonstrates a typical result. With a limited data size, the algorithm can produce a tree ensemble that converges well within $20$ iterations.

\begin{figure}[tb]
\begin{center}
\centerline{\includegraphics[width=0.9\columnwidth]{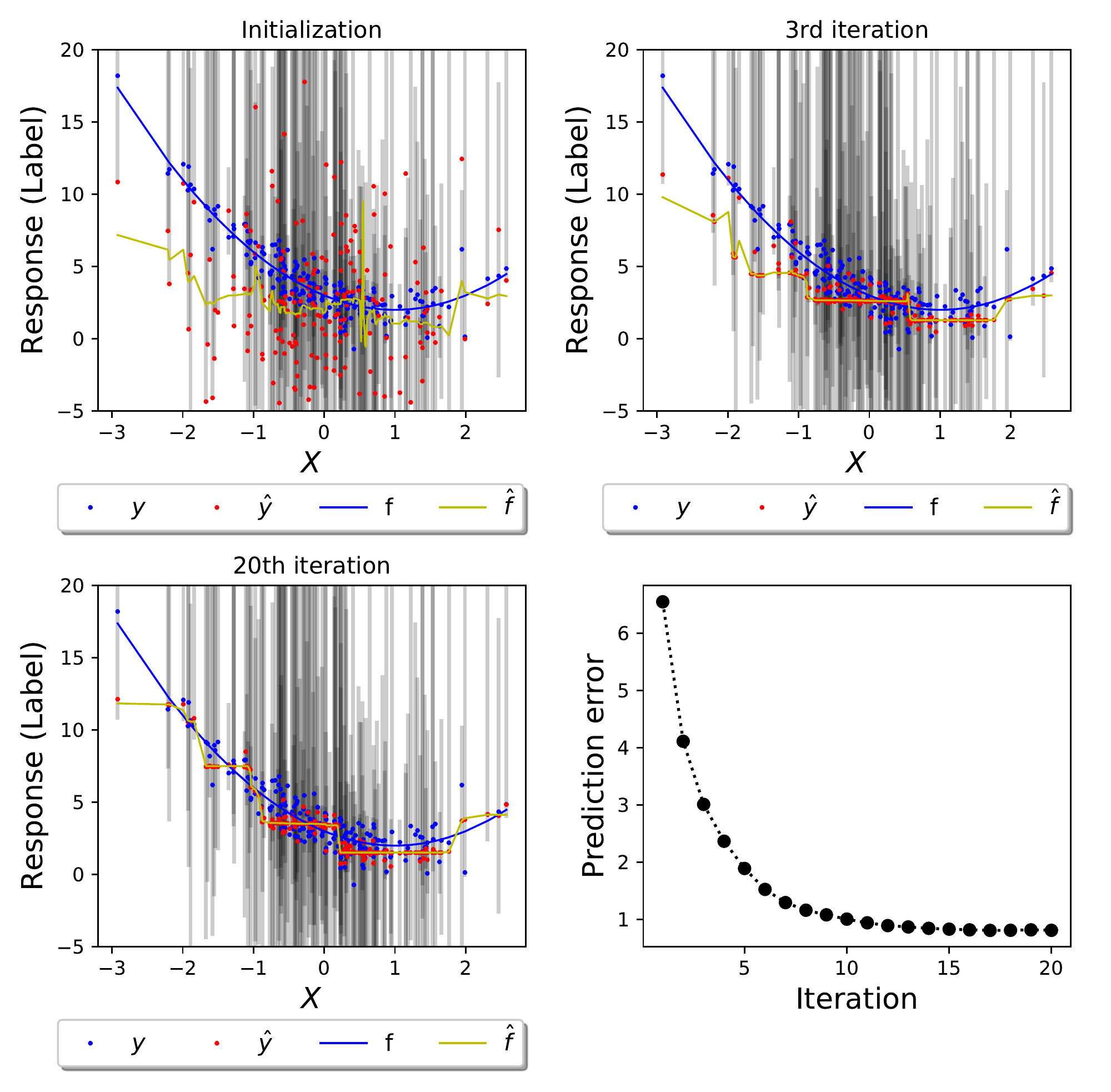}}
\vskip -0.1in
\caption{
Similar setting as Fig.\ref{fig_V4_linearReg} but for the nonlinear regression example.
}
\label{fig_V4_quadReg}
\end{center}
\vskip -0.2in
\end{figure}

\subsection{Case Study: Distribution from Individual-Adaptive Surveys} \label{subsec_data1}

We developed a web-based survey system and deployed it on MTurk to collect interval-private data. We de-identified the voluntary participants and randomized them into two groups. In the first group, each anonymous participant was asked privacy-sensitive questions in the form of Fig.~\ref{fig_example_age}(c), where the progressive mechanism was based on Example~\ref{eg_progression}. In particular, $[\underline{U}^{(0)}, \bar{U}^{(0)}]$ was specified according to the question and each $U^{(h)}$ was uniformly generated from $[\underline{U}^{(h-1)}, \bar{U}^{(h-1)}]$. We set a maximum number of three rounds so that each individual could answer from zero to three times. For comparison, the second group received conventional questions and submitted point-valued data. The sample size in each group was around $300$, and $Y$ represented pre-tax annual salary, available cash flow, and yearly frequency of intercourse (treated as continuous-valued variables, shown in Tab.~\ref{tab_Mturk}). 

We use the ultimate intervals from the progressive mechanism (`Round-X') to obtain the NPMLE $\hat{F}_n$. Using the empirical distribution of the point data collected from the second group to approximate the population $\Fy$, we calculate the squared Energy Distance $\textrm{ED}(\hat{F}_n,\Fy) \de \int_{\R}(\hat{F}_n(y)-\Fy(y))^2 dy$ as the estimation error, and record the privacy coverage. We repeat the above to only the data collected from the first round of questions (`Round-1'). As we can see from Tab.~\ref{tab_Mturk}, the progressive mechanism tends to reduce estimation error by adapting to individual-level privacy sensitivities.

\begin{table}[tb]
\centering
\caption{
Subsection~\ref{subsec_data1} experiment: Distribution estimation and coverage from non-progressive (`Round-1') and progressive interval mechanisms (`Round-X') for three questions.
}
\vspace{-0.2cm}
\label{tab_Mturk}
\scalebox{1}{
\begin{tabular}{cccccc}
\toprule
            & $\Y$             & \multicolumn{2}{c}{$\textrm{ED}(\hat{F}_n,F)$} & \multicolumn{2}{c}{Coverage} \\
            &                  & Round-1         & Round-X        & Round-1     & Round-X    \\ \midrule
Salary      & $[0,150]$ ($\times$\$1k) & $1.38$            & $0.71$           & $0.68$        & $0.35$       \\ \midrule
Cash        & $[0,300]$ ($\times$\$1k) & $1.72$            & $0.87$           & $0.75$        & $0.45$       \\ \midrule
Intercourse & $[0,100]$ ($\times$1)     & $0.73$            & $0.40$           & $0.89$        & $0.63$       \\ \bottomrule
\end{tabular}
}
\end{table}

\subsection{Case Study: Life Expectancy Regression} \label{subsec_data2}

In the experimental study, we considered the `life expectancy' data from the \textit{kaggle} open-source dataset~\cite{lifeData}, originally collected from the World Health Organization (WHO). 
The data consist of 193 countries from 2000 to 2015, with 2938 data items/rows uniquely identified by the country-year pair. The learning goal is to predict life expectancy using 20 potential factors, such as demographic variables, immunization factors, and mortality rates. 

We will exemplify the use of Algo.~\ref{algo1} under three mechanisms. 
The first mechanism (`Oracle') uses the raw data of $Y$ (life expectancy). 
The second mechanism (`$\mathcal{M}_1$') is described by $Y \mapsto Z$, where $Z$ is in the form of (\ref{eq_format}), $\Q=[U-1,U+1]$, $U$ is generated from the Logistic distribution with scale~$1$, and
	$\C=\{ y: U-1 < y \leq U+1 \}$. An {interpretation} is that during the data collection, individual data with an overly short or long life expectancy tend to be reported as half-interval (namely $\leq U-1$ or $>U-1$), while those within the mid-range $\C$ tend to be exactly reported. 
The third mechanism (`$\mathcal{M}_2$') is a Case-II mechanism described by (\ref{eq_format}), where $\Q=[U, V]$ is generated from the ordered Logistic distribution with scale~$2$.
The last mechanism (`$\mathcal{M}_3$') is a Case-I mechanism where $\Q=U$ is generated from the Logistic distribution with scale~$5$.
The {interpretation} of $\mathcal{M}_2$ or $\mathcal{M}_3$ is that individual data are quantized into random categories.
We calculate the privacy coverage (Definition~\ref{def_IP}) for each privacy mechanism using the empirical distribution and summarize it in Table~\ref{tab_real}.

For each mechanism, the predictive performance of the fitted regression under three methods, namely linear regression (LR), gradient boosting (GB), and random forest (RF), are evaluated using the five-fold cross-validation. The performance results are summarized in Table~\ref{tab_real}. The results show that a privacy mechanism with smaller privacy coverage tends to perform better, which is the expected phenomenon due to privacy-utility tradeoffs.  The results also show a (statistically) negligible performance gap between $\mathcal{M}_1$, $\mathcal{M}_2$, and the Oracle (meaning that the raw data are used). The performance starts to degenerate only in the last mechanism, where there is a large privacy coverage ($94\%$) or small privacy leakage ($6\%$).
To visualize the data and privacy coverage, we show a snapshot of the database in Tab.~\ref{fig_snap}, where we used the Case-II mechanism and generated $U, V$ from the standard Logistic distribution.

\begin{table}[tb]
\centering
\caption{Subsection~\ref{subsec_data2} experiment: Predictive performance of linear regression (LR), gradient boosting (GB), and random forest (RF) methods under different mechanisms, evaluated by the $R^2$ and mean absolute error (MAE) from 5-fold cross validations.
}
\vspace{-0.2cm}
\label{tab_real}
\scalebox{1}{
\begin{tabular}{cccccc}
\toprule
                    &          & Oracle        & $\mathcal{M}_1$ & $\mathcal{M}_2$ & $\mathcal{M}_3$ \\ \cline{2-6} 
                    & Coverage & $0\%$         & $57\%$          & $76\%$          & $94\%$          \\ \midrule
\multirow{2}{*}{LR} & $R^2$    & $0.79 (0.02)$ & $0.78 (0.02) $  & $0.78 (0.02) $  & $0.52 (0.09)$   \\
                    & MAE      & $3.21 (0.07)$ & $3.25 (0.08)$   & $3.23 (0.04)$   & $4.43 (0.25)$   \\ \midrule
\multirow{2}{*}{GB} & $R^2$    & $0.89 (0.01)$ & $0.86 (0.01)$   & $0.85 (0.01) $  & $0.74 (0.01)$   \\
                    & MAE      & $2.3 (0.18)$  & $2.56 (0.11)$   & $2.68 (0.09)$   & $3.65 (0.13)$   \\ \midrule
\multirow{2}{*}{RF} & $R^2$    & $0.82 (0.02)$ & $0.78 (0.02)$   & $0.75 (0.02)$   & $0.69 (0.02)$   \\
                    & MAE      & $2.86 (0.12)$ & $3.25 (0.09)$   & $3.41 (0.14)$   & $3.97 (0.05)$   \\ \bottomrule
\end{tabular}
}
\end{table}

\section{Related Literature} \label{sec_literature}

This section reviews other perspectives on data privacy. 

\textit{Database privacy}.
A popular way of evaluating data privacy is through \textit{differential privacy}~\cite{dwork2006calibrating}, a cryptographically motivated definition to protect the existence of an individual identity in a database~\cite{chaudhuri2011differentially,sarwate2013signal,dong2019gaussian,neunhoeffer2020private,vietri2020new}. 
A database is a matrix whose rows represent individuals and columns represent their attributes.
Differential privacy measures privacy leakage by a parameter $\v$ that bounds the likelihood ratio of the output of an algorithm under two databases differing in a single individual. The standard tool for creating differential privacy is the sensitivity method~\cite{dwork2006calibrating}, which first computes the desired algorithm output from the database, and then adds noise proportional to the largest possible change induced by modifying a single row in the database.
Differential identifiability~\cite{lee2012differential} was developed as an alternative formulation to guarantee differential privacy, based on the probability of individual identification conditional on the output. This notion was also extended to the identifiability of databases~\cite{wang2016relation}.

\textit{Sanitization}. In some applications, one must publish an anonymized and perturbed version of the original database (also known as `sanitization') to protect individual privacy. In this direction, a classical approach is based on the notion of $k$-anonymity~\cite{sweeney2002k}, meaning that for every individual, there exist $k-1$ others with the same tuple of non-private attribute values (assumed to exist) for a pre-specified $k$. Since $k$-anonymity does not necessarily protect private attributes, there have been extensions such as the $t$-closeness~\cite{li2007t}.
Moreover, information-theoretic quantities such as mutual information and average distortion have been used to quantify privacy in database sanitization (see, e.g.,~\cite{rebollo2009t,sankar2013utility,makhdoumi2013privacy,wang2016relation}).

\textit{Local privacy}.
The main difference between database privacy and local privacy (the focus of this work) is summarized below. First, local privacy protects each data value or the associated individual identity during data collection, while database privacy protects the presence of an individual in an already-collected database. Second, local privacy is supposed to disclose or collect individual-level data, while database privacy is developed for querying summary statistics. Third, in practical implementations, database privacy involves three parties: data owners (individuals), a data collector (trusted third party, often an organization) who maintains the database, and analysts who query statistics from the database. On the other hand, local privacy may only involve data owners and an (untrusted) data collector who may immediately analyze the collected data. Local privacy is much less studied in the literature than database privacy. Interval privacy can be regarded as a framework for local privacy.

\textit{Local differential privacy}.
An existing notion of local privacy is local differential privacy~\cite{evfimievski2003limiting,kasiviswanathan2011can,sarwate2014rate}. 
It restricts the conditional distributions of the privatized data on any two different raw data to have a density ratio close to one, often realized by perturbing the raw data with additive noise.  
We illustrate the difference between interval privacy and local differential privacy through an example. 
Consider a salary of \$25k and another salary of \$250k. The two private values will be perturbed into two random variables with similar densities in differential privacy.  
In contrast, under interval privacy, the two private values are obfuscated with two random intervals, say (0, \$100k) and [\$200k, $\infty$).
An operational difference is that interval privacy offers information by narrowing down the support, while local differential privacy offers information by perturbing the value. 
A conceptual difference is that interval privacy ensures an adversary does not gain additional information of $Y$ on large support based on its collected data $Z$ and prior knowledge on $Y$ (through posterior ratio of $Y \mid Z$), while local differential privacy limits the additionally gained information through the likelihood ratio of $Z \mid Y$.

\section{Conclusion and Further Remarks} \label{sec_con}

We developed the concept, theory, and use scenarios of interval privacy.
Here are some distinct challenges that we address while the existing privacy may not be good. \\
\indent 1) \textit{Transparency}: individuals can easily perceive the collected data. Existing local privacy is often practically operated by a data collector. Consequently, there may be an abuse of privacy budgets not known to those who have submitted data in the first place. In contrast, interval privacy can be naturally operated as an interface where each individual can perceive the collected data. Such transparency is crucial for local privacy, where the data collector is untrustworthy. \\
\indent 2) \textit{Flexibility}: interval privacy allows progressive refining of collected information and thus can be adaptive to individuals' heterogeneous privacy sensitivities. From a data collector's perspective, such flexibility tends to enhance the quality of information compared with using a fixed budget. \\
\indent 3) \textit{Fidelity}: individuals can submit authentic information with goodwill, and a data collector can interpret or analyze data without factual errors. Information fidelity can be indispensable in application domains such as census, security, and defense. For example, in a demographic study where scientists collect age information from a cohort of residents, they may find that at least $80\%$ residents are below $30$ years old and publish that population-level fact. It is difficult for scientists to make such a statement without information fidelity. 

We mention some potential future work.
First, it is worth extending interval privacy from continuous-valued variables to discrete ones, e.g., categorical, ordinal, and count data. 
Second, it is interesting to apply interval privacy to supervised, unsupervised, and collaborative learning (e.g.,~\cite{DingFL,DingRecommender}) where learners share interval-private statistics.
Third, analyzing privacy-utility tradeoffs in various interval mechanisms deserves further study.
The \textit{Appendix} and \textit{Supplementary Document} contain further discussions and technical details.

\appendix

\section*{Proof of Theorem~\ref{thm_cover}}

	We first consider the Case-I mechanism.
	By our definition, 
	\begin{align}
		\tau(\M)
		&=\E \biggl[ \E \bigl\{\i_{Y \leq U} \Fy(U)+\i_{Y > U} (1-\Fy(U)) \bigr\} \big | U\biggr] \nonumber \\
		&=	\E \biggl[ \Fy(U)^2 + (1-\Fy(U))^2 \biggr] 
		\geq \frac{1}{2} , \label{eq_41}
	\end{align}
	where the last line is by the Cauchy's inequality.
	For any $\tau \in (1/2,1)$, we will prove the existence of the density of $U$ such that $\tau(\M)=\tau$ by construction.
	For parameters $\pi_1 \in [0,1], \pi_2=1-\pi_1, \sigma>0$, we let the density of $U$ be 
	\begin{align}
	p_{\pi_1,\pi_2,\sigma}(u) =\frac{1}{\sqrt{2\pi}\sigma} \biggl\{ \pi_1 e^{-\frac{(u-\mu_{\v})^2}{2\sigma^2}} +\pi_2 e^{-\frac{(u-\mu_{1})^2}{2\sigma^2}} \biggr\} , \nonumber
	\end{align}
	where $\mu_{1}=F^{-1}(1/2)$ and $\mu_{\v}$ is to be selected.
	The above density naturally induces the function $h: [\pi_1,\sigma] \mapsto \tau(\M)$.
	
	We first prove that for any small $\v \in (0,1/2)$, there exist $\pi_1,\sigma$ such that $\tau(\M) > 1-\v$.
	Since $\Fy(u)^2 + (1-\Fy(u))^2$ is nonincreasing for $u<\mu_1$ and it approaches $1$ as $u \rightarrow -\infty$, and $d\Fy(u)/du$ is bounded, there exists a $\mu_\v$ with $\mu_\v<\mu_1$ such that $\Fy(u)^2 + (1-\Fy(u))^2 > 1-\v/2$ for all $u$ in a neighborhood of $\mu_\v$. Then a $\pi_1$ close to one and a $\sigma$ close to zero ensures that
	$ \int_{\mu\in \R} (2\pi\sigma^2)^{-1/2} \pi_1 e^{-(u-\mu_{\v})^2/(2\sigma^2)} \geq 1-\v,
	$
	which further implies that $h(\pi_1,\sigma)=\tau(\M)\geq 1-\v$.
	By a similar argument, we can prove that for any small $\v \in (0,1/2)$ there exists $\pi_1,\sigma$ such that $h(\pi_1,\sigma) \leq 1/2+\v$. 
	For any $\tau \in (1/2,1)$, the above arguments show the existence of two sets of $[\pi_1,\sigma]$ so that $h(\pi_1,\sigma)$ sandwiches $\tau$. By the continuity of $h$, we conclude the existence of  $[\pi_1,\sigma]$ such that $h(\pi_1,\sigma) = \tau$.
	
	The proof for Case-II interval mechanism follows from  
	\begin{align}
		\tau(\M)
		&=	\E \bigl[ \Fy(U)^2 + (\Fy(V)-\Fy(U))^2+(1-\Fy(V))^2 \bigr] \nonumber 
	\end{align}
	and similar arguments used in the above Case-I.

\section*{Proof of Theorem~\ref{thm_optimalU}}

Let $p_U$ denote the density of the distribution of $U$.
Standard results~\cite{groeneboom1992information} show that an explicit expression of the information lower bound for Case-I is given by 
$
\int_{\Y} \bigl(\frac{d}{d y}\phi(y)\bigr)^2 \Fy(y)(1-\Fy(y))/p_U(y) dy
$
and that the NPMLE attains the lower bound.
By the Cauchy's inequality, 
\begin{align*}
	&\int_{\Y} \biggl(\frac{d}{d y}\phi(y)\biggr)^2 \frac{\Fy(y)\{1-\Fy(y)\}}{p_U(y)} dy \\
	&=\biggl(\int_{\Y}p_U(y) dy\biggr) \cdot \int_{\Y} \biggl(\frac{d}{d y}\phi(y)\biggr)^2 \frac{\Fy(y)\{1-\Fy(y)\}}{p_U(y)} dy 	\\
	&\geq \int_{\Y} \biggl|\frac{d}{d y}\phi(y)\biggr| \sqrt{ \Fy(y)\{1-\Fy(y)\}} dy,
\end{align*}
with equality when
$
	p_U(y) \propto \bigl|\frac{d}{d y}\phi(y)\bigr|  \sqrt{ \Fy(y)\{1-\Fy(y)\}},
$
given that it is integrable.

\section*{Proof of Theorem~\ref{thm_validity}}

Given $Z=\{\Q,I(\Q,Y),Y \i_{Y \in \C}]$, using Bayes' theorem and the independence between $Y$ and $\Q$, we have 
	$p(y \mid Z) 
	= p(y \mid \Q, I_{\Q,y}, Y  \i_{Y \in \C})
	= c \cdot p(I_{\Q,y} , \i_{y \in \C}\mid y, \Q) \cdot p(y \mid \Q)
	= c \cdot p(I_{\Q,y}, \i_{y \in \C} \mid y, \Q) \cdot p(y)
	= c \cdot  \i_{y \in S_Z} p(y)
	$ if $y \not\in \C$, and $p(y \mid Z)=\delta(y)$ otherwise,
	where $S_Z = \Ra^{(j)} - \C$ is the difference set, $\Ra^{(j)}$ is the interval that $\Q$ and $y$ determine, and $c$ is a constant that not depending on $y$. 
	This concludes the proof.

\section*{Proof of Theorem~\ref{thm_ensemble}}

We use the following lemma, proved in the supplement.

\begin{lemma}\label{lemma1}
	Suppose that $x_1,\ldots,x_m$ are nonnegative values that sum to one. 
	Suppose that $S_k,k\in [1:K]$ and $R_j,j\in [1:J]$ are two partitions of $\Omega=\{1,\ldots,m\}$ such that the intersection of $S_k$ and $R_j$ contains at most one element for any $k,j$.
	Then,
	\begin{align*}
		\biggl( 1-\sum_{k=1}^K \biggl(\sum_{i \in S_k} x_i\biggr)^2\biggr)
		+\biggl( 1-\sum_{j=1}^J \biggl(\sum_{i \in R_j} x_i\biggr)^2\biggr)
		\geq 1-\sum_{i=1}^m x_i^2 .
	\end{align*}

\end{lemma}

Next, we prove Theorem~\ref{thm_ensemble}.
	We first consider the case $\C=\C_{[1]}\cup\C_{[2]}=\emptyset$ so  the data are always in the form of intervals.
	We only need to prove the result for any two mechanisms $\M_1$ and $\M_2$, namely 
	$
		1-\tau\bigl(\oplus_{j=1}^2 \M_j\bigr) 
		\leq \sum_{j=1}^2 \bigl(1-\tau(\M_j)\bigr). 
	$
	The result for multiple mechanisms will then follow from induction.
	For a mechanism $\M$ with anchors $-\infty=\Q^{(0)},\ldots,\Q^{(m)}=\infty$, by a similar argument as (\ref{eq_41}), 
	\begin{align}
		\tau(\M)= \E\biggl(\sum_{i=1}^m \P_Y(\Q^{(i)})^2\biggr),\label{eq_52}
	\end{align}
	where $\Q^{(i)} \de (\Q^{(i-1)}, \Q^{(i)}]$, $\P_Y$ is the probability function of $Y$, and the expectation is over $\Q^{(i)}$'s.
 
	Suppose that the intersections of $\M_1,\M_2$ produce a finer set of intervals $I_{1},\ldots,I_{m}$, and each has a probability measure $x_i=\P_Y(I_{i})$, $i\in [1:m]$. 
	Let $\Omega=\{1,\ldots,m\}$.
	Suppose that the intervals under $\M_1$ (respectively $\M_2$) correspond to $S_k,k\in [1:K]$ (respectively $R_j,j\in [1:J]$) which partitions $\Omega$.
	To prove the theorem, according to (\ref{eq_52}), it suffices to prove for each outcome of $\{I_1,\ldots,I_m\}$ that
	\begin{align}
	 &\biggl\{1-\sum_{k=1}^K \biggl(\sum_{i\in S_k} \P_Y(I_i)\biggr)^2 \biggr\}+
	 \biggl\{1-\sum_{j=1}^J \biggl(\sum_{i\in R_j} \P_Y(I_i)\biggr)^2 \biggr\}\nonumber  \\
	  &\geq 1-\sum_{i=1}^m \P_Y(I_i)^2 . \nonumber
	\end{align}
	Also, by the definition of the set systems  $\{S_k\}_{k=1}^K$ and $\{R_j\}_{j=1}^K$, the intersection of $S_k,R_j$ for any $k,j$ contains at most one element.
	The proof thus follows from Lemma~\ref{lemma1}. 
	
	For the case $\C \not=\emptyset$, we will use the above-proved result.
	In particular, let $\tilde{\tau}(\M)$ denote the privacy coverage for $\M=\M_1 \oplus \M_2$ if $\C$ were hypothetically set to be empty (i.e. the mechanism that is fully interval-valued).
	Let $\textrm{FI}(\C), \textrm{FI}(\C_{[1]}), \textrm{FI}(\C_{[2]})$ denote the finest intervals in $\C$, $\C_{[1]}$, $\C_{[2]}$, respectively. 
	Then, for a realization of the anchors $\Q$,  
	$
		\tilde{\tau}(\M) = \tau(\M) + \sum_{I \in \textrm{FI}(\C)} \P(I)^2.
	$
	Similar identities hold for $\M_1$ and $\M_2$.
	We already proved 
	\begin{align}
		1-\tilde{\tau}(\M) \leq (1-\tilde{\tau}(\M_1)) + (1-\tilde{\tau}(\M_2)),\label{eq_100}
	\end{align}
	and due to $\C=\C_{[1]} \cup \C_{[2]}$, we also have  
	\begin{align}
		\sum_{I \in \textrm{FI}(\C)} \P(I)^2 
		\leq & \sum_{I \in \textrm{FI}(\C_{[1]})} \P(I)^2 
		+ \sum_{I \in \textrm{FI}(\C_{[2]})} \P(I)^2. \label{eq_101}
	\end{align}
	We conclude the proof by combining (\ref{eq_100}) and (\ref{eq_101}).

\section*{Proof of Theorem~\ref{thm_transform}}

	Suppose that $I(\Q,Y) = j$, or equivalently $g(Y) \in \Ra^{(j)}$. If $g(Y) \not\in \C_g$, then the corresponding $Z$ satisfies
	\begin{align*}
		L(S_{Z}) = \P( g(Y) \in \Ra^{(j)}) = \P( Y \in g^{-1}(\Ra^{(j)}))
		= L(S_{Z_g}).
	\end{align*}
	Meanwhile, because $Y \in \C$ implies $g(Y) \in \C_g$, the probability of $g(Y)$ falling into  $\C_g$ (which results in a zero-size) is not smaller than that of $Y$ falling into $\C$. 
 	Thus, by the definition of $\tau(\cdot)$ we have $\tau(\M_g) \geq \tau(\M)$. The equality holds if and only if $\P_Y( g(Y) \in \C_g) = \P_Y( Y \in \C)$, namely $L(g^{-1}(\C_g)) = L(\C)$.

\section*{Proof of Theorem~\ref{thm_post}}

From the Bayes' theorem and Markovity $Y \rightarrow Z \rightarrow W$,
	\begin{align}
		&p(y \mid Z=z, W=w) 
		= c_1 p(z, w \mid y) p(y) \nonumber \\
		&= c_1  p(z \mid y) p(w \mid z, y) p(y) 
		=  c_1  p(z \mid y) p(w \mid z) p(y) \nonumber \\
		&= c_1 p( \q, I_{\q,y}, y \i_{y \in \C} \mid y) p(w \mid z) p(y) \nonumber \\
		&= \left\{
		\begin{array}{lcl}
		\delta(y)       &      & {\textrm{ if }y  \in \C}\\
		 c_1 p( I_{\q,y}, y \i_{y \in \C} \mid \q, y) p(\q) p(w \mid z) p(y) &&\\
		= c_1 \i_{y \in S_z} p(q) p(w \mid z) p(y)  &      & {\textrm{ if }y \not\in \C}
		\end{array} \right. \nonumber
	\end{align}
	where $c_1$ is a constant that does not depend on $y$, and $S_z$ is the interval that $q$ and $y$ uniquely determine.
	This implies $\tau$-interval privacy by Definition~\ref{def_IP}. 

\section*{Proof of Proposition~\ref{prop_eg2}}

	It can be calculated that for each $i$,
	\begin{align*}
		&\E\bigl(  \Delta_i (2U_i-b) + (1-\Delta_i) (2U_i-a) \mid Y=Y_i \bigr) \\
		&= \int_{Y_i}^{b} \frac{1}{b-a} \cdot (2u-b) du
		+\int_{a}^{Y_i} \frac{1}{b-a} \cdot (2u-a) du 
		= Y_i .
	\end{align*}
	Thus, by the i.i.d. assumption, 
	$$\E(\hat{\mu}_n) = \E \bigl( \Delta_1 (2U_1-b) + (1-\Delta_1) (2U_1-a)  \bigr) = \E(Y_1) = \mu.$$
	The boundedness of $X_i$ and $U_i$ implies $\var(\hat{\mu}_n)  = O(n^{-1})$.

\section*{Proof of Theorem~\ref{thm_NPMLE}}

	Since $\hat{F}_n$ is an NPMLE, for each $\v \in(0,1)$, we have $\lim_{\v\rightarrow 0^+} \v^{-1} \psi((1-\v)\hat{F}_n + \v \Fy) - \psi(\hat{F}_n) \leq 0$, implying
	\begin{align}
		&\int_{\R^q \times \Y} \sum_{i=1}^m \B^{(i)} \frac{\Fy(\Ra_t^{(i)}) - \hat{F}_n(\Ra_t^{(i)})}{\hat{F}_n(\Ra_t^{(i)})} \, d \P_n(t,y) \leq 0, \textrm{ namely} \nonumber \\
		&\int_{\R^q \times \Y} \sum_{i=1}^m \frac{\B^{(i)} \Fy(\Ra_t^{(i)})}{\hat{F}_n(\Ra_t^{(i)})} \, d \P_n(t,y) \leq 1 . \label{eq_201}
	\end{align}
	Let $\P, \Omega, \omega$ denote the probability measure, sample space, and an outcome of (infinite) sequences $[T_1,Y_1], \ldots, [T_n,Y_n]$, respectively. By the strong law of large numbers, $\P_n(\cdot,\cdot,\omega)$ converges weakly to $\P(\cdot,\cdot,\omega)$ for all $\omega$ in a set with one $\P$-measure. 
	Fix $\v\in(0,1/q)$ and define
	$
		A_{\v} \de \{t: \Fy(\Ra_t^{(i)})\geq \v\}.
	$
	By the convergence of $\P_n(\cdot,\cdot,\omega)$ to $\P(\cdot,\cdot,\omega)$, there exists a constant $C>0$ such that
	\begin{align}
		\sum_{i=1}^m 1/\hat{F}_n(\Ra_t^{(i)}, \omega) \leq C \label{eq_C}
	\end{align}
	for $t \in A_{\v}$ and all sufficient large $n$.
	By the Helly's selection theorem, the sequence $\{\hat{F}_n(\cdot, \omega)\}_n$ has a subsequence $\{\hat{F}_{n_k}(\cdot, \omega)\}_n$, converging vaguely to a non-decreasing right-continuous function that takes values in $[0,1]$, denoted by $F$. 
	Next, we use the following lemma, proved in the supplement.
	
	\begin{lemma}\label{lemma_NPMLE}
		With Monotonicity condition and Inequality (\ref{eq_C}), 
		\begin{align}
			&\lim_{k \rightarrow \infty} \int_{A_{\v} \times \Y} \sum_{i=1}^m \frac{\B^{(i)} \Fy(\Ra_t^{(i)})}{\hat{F}_{n_k}(\Ra_t^{(i)})} \, d \P_{n_k}(t,y) \nonumber \\
			&=\int_{A_{\v} \times \Y} \sum_{i=1}^m  \frac{\B^{(i)} \Fy(\Ra_t^{(i)})}{F(\Ra_t^{(i)})} \, d \P(t,y) \label{eq_lemma} .
		\end{align}	
	\end{lemma}
	It follows from Inequality (\ref{eq_201}) and Lemma~\ref{lemma_NPMLE} that the right-hand side in (\ref{eq_lemma}) is not larger than one.	
	Consequently, applying the monotone convergence theorem, we have
		\begin{align}
			&\int_{\R^q \times \Y} \sum_{i=1}^m \frac{\B^{(i)} \Fy(\Ra_t^{(i)})}{F(\Ra_t^{(i)})} \, d \P(t,y) \nonumber \\
			&=\lim_{\v \rightarrow 0^+}\int_{A_{\v} \times \Y} \sum_{i=1}^m \frac{\B^{(i)} \Fy(\Ra_t^{(i)})}{F(\Ra_t^{(i)})} \, d \P(t,y)
			\leq 1. \label{eq_202}
		\end{align}	
	Meanwhile, since $\int_{\Y}\B^{(i)}d \P(y)=\Fy(\Ra_t^{(i)})$, we have
	\begin{align}
		&\int_{\R^q \times \Y} \sum_{i=1}^m  \frac{\B^{(i)} \Fy(\Ra_t^{(i)})}{F(\Ra_t^{(i)})} \, d \P(t,y) 
		=\int_{\R^q} \sum_{i=1}^m  \frac{\Fy^2(\Ra_t^{(i)})}{F(\Ra_t^{(i)})} \, d \P(t) \nonumber \\
		&\geq \int_{\R^q} \frac{\bigl\{\sum_{i=1}^m\Fy(\Ra_t^{(i)})\bigr\}^2}{\sum_{i=1}^m F(\Ra_t^{(i)})} \, d \P(t) 
		= \int_{\R^q} 1 \, d \P(t)  = 1, \label{eq_cauchy} 
	\end{align}	
	where (\ref{eq_cauchy}) is from the Cauchy's inequality. Then, it follows from (\ref{eq_202}) that the equality in (\ref{eq_cauchy}) must hold, which implies that $F(\Ra_t^{(i)})=\Fy(\Ra_t^{(i)})$ for all $i \in [1:m]$ and $t$ with a positive density.
	Combining this and the Resolvability condition, we have that for each $y$ in the closure of $\Y$, there exists an open neighborhood $N(y)$ such that for all $y_a,y_b\in N(y)$, $F((y_a,y_b])=\Fy((y_a,y_b])$. Combining this and the finite cover theorem, for any constants $c$ and $c'$ such that $(c,c']$ is in the closure of $\Y$, we have finitely many numbers $c=c_0 < c_1< \cdots < c_H=c'$ such that $F((c_{h-1},c_h])=\Fy((c_{h-1},c_h])$ for $h\in [1:H]$, implying $F((c,c'])=\Fy((c,c'])$. Thus, $F=\Fy$.
	
	Therefore, for all $\omega$ in a set with one $\P$-measure, each subsequence of $\{\hat{F}_n(\cdot, \omega)\}_{n}$ has a convergence subsequence, and they all have the same limit $\Fy$. This implies that $\{\hat{F}_n(\cdot, \omega)\}_{n}$ converges weakly to $\Fy$ with $\P$-probability one. Since $\Fy$ is continuous, we further conclude Theorem~\ref{thm_NPMLE}.

\section*{Proof of Proposition~\ref{prop_ignorability}}

Let $E$ and $E'$ denote the events of collecting null and $L(S_Z) < \tau$. 
By Bayes' theorem, $\P(E \mid Z=0)=1$, and the assumption $\P(Z=0)=1-\rho \geq \P(E')$, we have
\begin{align*}
	\P(Z=0 \mid E)
	&=\frac{\P(Z=0)}{\P(Z=0)+\P(E \mid Z=1,E')\P(Z=1,E')} \\
	&\geq \frac{\P(Z=0)}{\P(Z=0)+\P(E')} \geq 0.5.
\end{align*}

\section*{Proof of Theorem~\ref{prop_reg}}

	The proof uses a similar technique in proving the consistency of classical maximum likelihood estimators.
	For notational convenience, let $\ell(f) = \E(\tilde{Y}-f(X))^2$.
	For an arbitrary $\v>0$, we will prove that 
	$
		\P( \norm{\hat{f}_n - f^*}_{\P_X} \geq \v) \rightarrow 0 
	$
	as $n\rightarrow \infty$.
	By the definition of $\hat{f}_n$, we have
	$$
		\E_n(\tilde{Y}-\hat{f}_n(X))^2 \leq 
		\E_n(\tilde{Y}-f^*(X))^2
		= \ell(f^*) + o_p(1)
	$$
	where the last equality is implied by the assumption.
	Therefore, the assumption further implies that
	\begin{align}
		&\ell(\hat{f}_n) - \ell(f^*)
		\leq 
		\ell(\hat{f}_n) - \E_n(\tilde{Y}-\hat{f}_n(X))^2 +o_p(1) \nonumber \\
		&\leq \sup_{f\in \F}\bigl|\E_n(\tilde{Y}-f(X))^2 - \ell(f))^2 \bigr|+o_p(1) \limp 0 \label{eq54}
	\end{align}
	as $n \rightarrow \infty$.	
	We rewrite $\E\{\tilde{Y}-f(X)\}^2$ as 
	\begin{align*}
		& \var(\tilde{Y}-f(X)) + \{\E(\tilde{Y}-f(X))\}^2 \\
		&= \var(\tilde{Y}-f(X)) + \{\E(f^*(X)-f(X))\}^2 \\
		&=\E( \var(\tilde{Y}-f(X) \mid X)) \\
		&\quad + \var(\E (\tilde{Y}-f(X) \mid X)\} + \{\E(f^*(X)-f(X))\}^2  \\
		&= c + \var(f^*(X)-f(X)) + \{\E(f^*(X)-f(X))\}^2  \\
		&= c + \E(f^*(X)-f(X))^2
		= c + \norm{f^*-f}^2_{\P_X}
	\end{align*}
with $c \de E( \var(\tilde{Y} \mid X ))$,
	which implies that
	$ 
		\inf_{f\in \F: \norm{f-f^*}_{\P_X} \geq \v }\ell(f) > \ell(f^*). 
	$ 
	This inequality ensures that there exists $\eta >0 $ such that $\ell(f) \geq \ell(f^*) + \eta$ for all $f$ satisfying $\norm{f^*-f}_{\P_X} \geq \v$.
	Therefore, 
	$ 
		\P(\norm{\hat{f}_n - f^*}_{\P_X} \geq \v) 
		\leq \P(\ell(\hat{f}_n) \geq \ell(f^*) + \eta ),
	$ 
	which, according to (\ref{eq54}), further goes to zero as $n \rightarrow \infty$.	
	This concludes the proof.

\section*{Acknowledgements}
\noindent We thank Xuan Bi, Robert Calderbank, Yuejie Chi, Ruobin Gong, Xinran Wang, Steven Wu, and Yu Xiang for their helpful discussions. 

\ifCLASSOPTIONcaptionsoff
  \newpage
\fi

\balance
\bibliography{privacy,J}
\bibliographystyle{IEEEtran}

%
%

\clearpage
\onecolumn
\setcounter{page}{1}
\centerline{\LARGE Supplementary Document for Interval Privacy}

\vspace{1cm}

The supplementary document includes the following sections.

\begin{itemize}
	\item S1. Further Experimental Studies
	\item S2. Computation of Conditional Means in Algo.~\ref{algo1}
	\item S3. Further Discussions on Related Work
	\item S4. Proof of Lemma~\ref{lemma1}
	\item S5. Proof of Lemma~\ref{lemma_NPMLE}
	\item S6. Distributional Identifiability of Selective Mechanisms
	\item S7. Additional Remarks
\end{itemize}

\section*{S1. Further Experimental Studies} 

We include two additional experimental studies in this section.

\subsection*{Tradeoff between Learning and Privacy Coverage}

The tradeoff between privacy coverage and learning performance is computable often in parametric settings, where the asymptotic variance and coverage privacy can be treated as functions of distribution parameters, and in some nonparametric learning contexts (see, e.g., Theorem~\ref{thm_optimalU} and relevant discussions).
In an experiment, we demonstrate the tradeoff with $n=200$ data as used in the first experiment of Subsection~\ref{subsec_reg_exp}.
We consider the Case-I mechanism, where $U$ is generated from Logistic distributions with scales $0.1, 0.3, 0.5, 1, 3, 5, 10, 20$, and $30$. We numerically compute the prediction errors and privacy coverages. The results, summarized in Fig.~\ref{fig_V5_tradeoff}, indicate that the performance is not sensitive to privacy coverage unless the latter is very close to one.

\begin{figure}[h!]
\begin{center}
\centerline{\includegraphics[width=0.7\columnwidth]{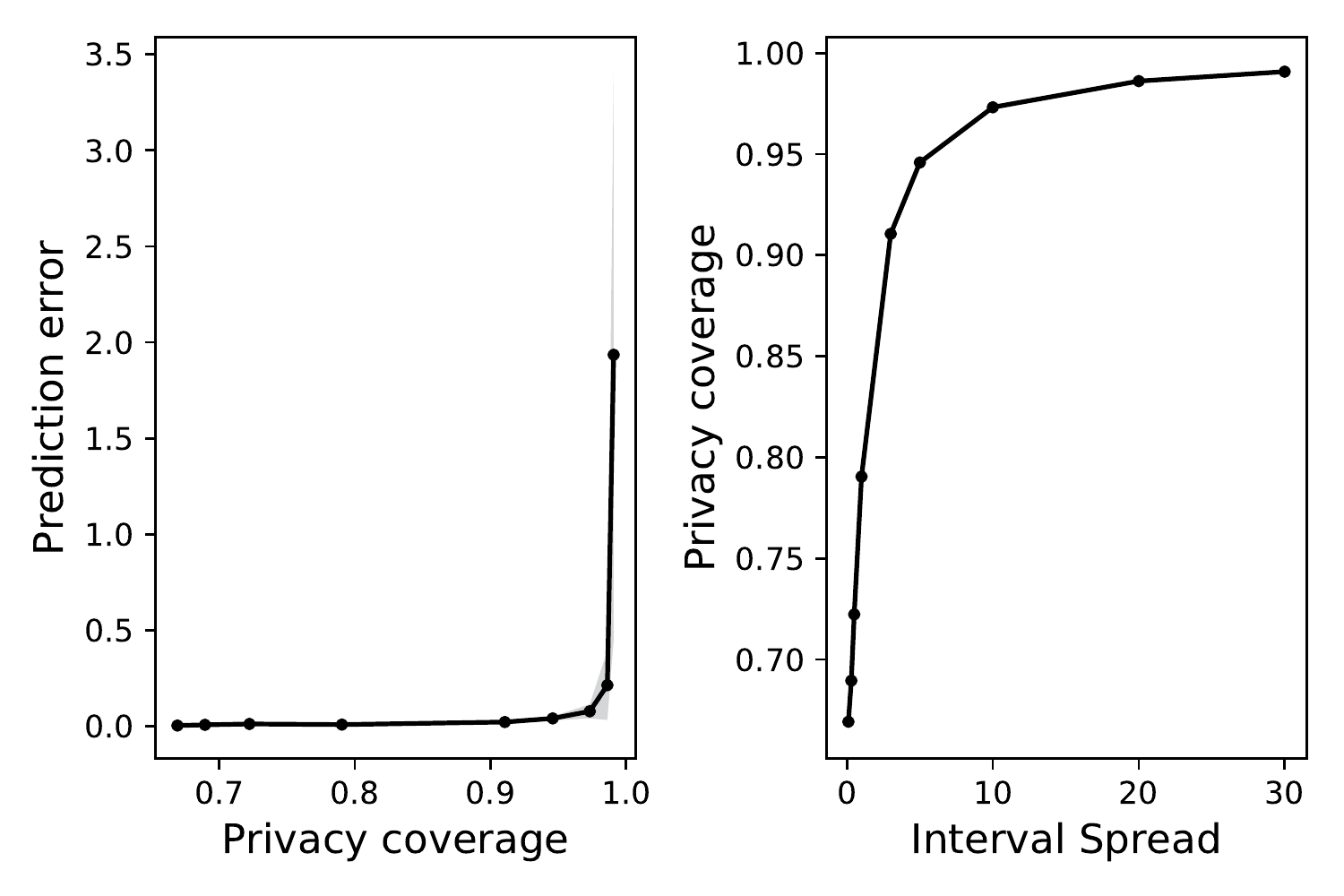}}
\vspace{-0.4cm}
\caption{The prediction error versus privacy coverage (left), and privacy coverage versus the spread of intervals, as measured by the standard deviation of $U$ (right). The shaded bands indicate $\pm$standard errors from 50 replications.}
\label{fig_V5_tradeoff}
\end{center}
\vskip -0.1in
\end{figure}

\subsection*{Sensitivity of Misspecified Noise}

\begin{figure}[h!]
\begin{center}
\centerline{\includegraphics[width=0.5\columnwidth]{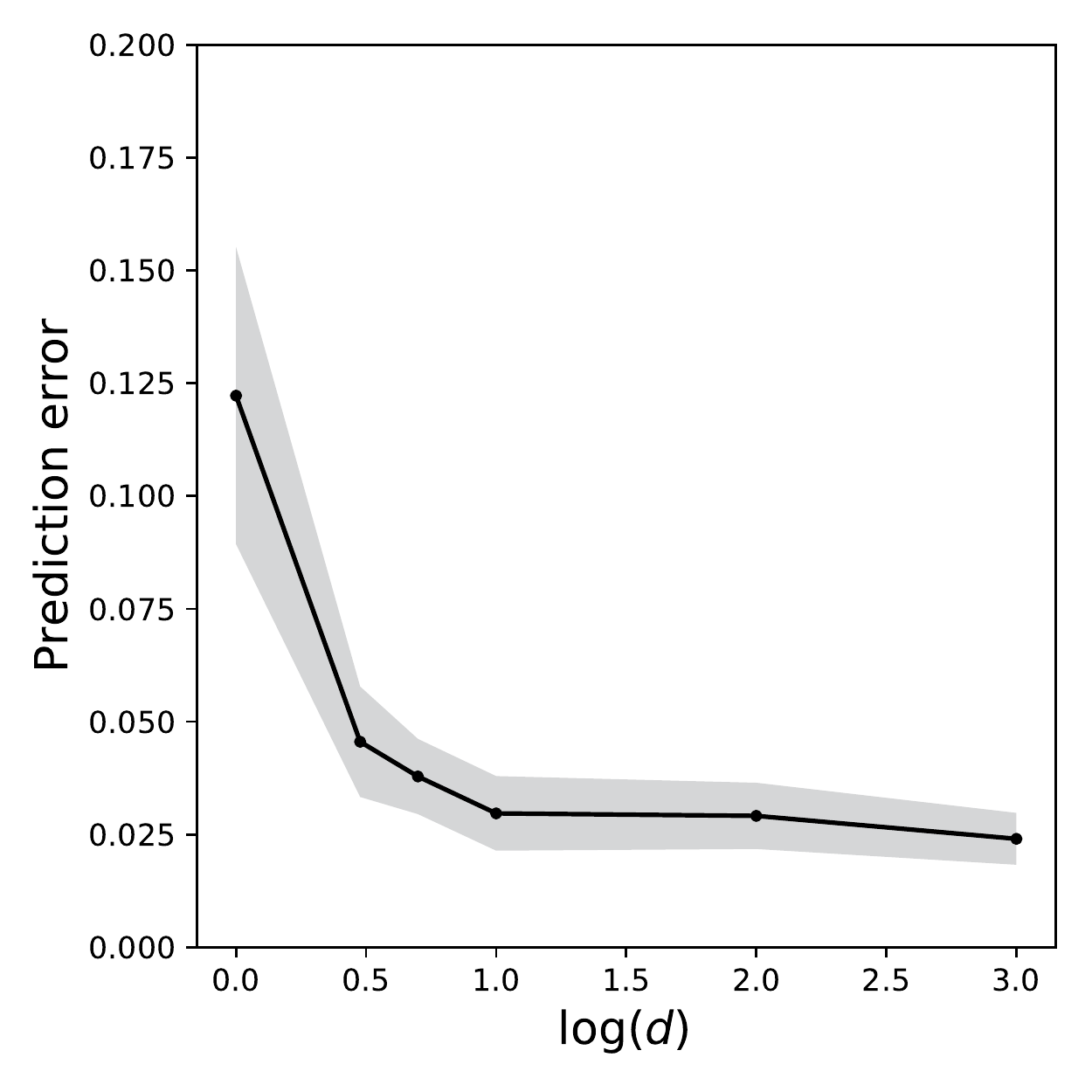}}
\vskip -0.2in
\caption{Performance (mean squared error) versus misspecification level (in terms of the $t$-degree of freedom $d$). A larger $d$ means less misspecification. The bands indicate $\pm$standard errors from 200 independent replications.}
\label{fig_V6_sensitivity}
\end{center}
\vskip -0.2in
\end{figure}
We empirically found that the estimation accuracy is generally not much affected by a  misspecified distribution of $\v$ when calculating (\ref{eq22}).
We demonstrate the sensitivity of wrongly specifying a distribution term using a specific example. 
A more sophisticated sensitivity analysis is left as future work.
We generate data in the same way as in Subsection~\ref{subsec_reg_exp}, except that the actual noise follows \textit{t}-distributions with degrees of freedom $d=1000, 100, 10, 5, 3$, and $1$. Here, $d=1000$ is virtually Gaussian while $d=1$ corresponds to a (heavy-tailed) Cauchy distribution. The postulation is still a Gaussian noise (so that it is misspecified).
The results summarized in Fig.~\ref{fig_V6_sensitivity} indicate that the performance (evaluated by the mean squared error) is not severely affected, and less deviation tends to produce less degradation in performance.

\section*{S2. Computation of Conditional Means in Algo.~\ref{algo1}}

We take the Case-II interval mechanism as an example.
Recall that $\v$ has CDF $F_{\v}$. 
We let $\G(s) = \int_{-\infty}^{s} x d \Fv(x)$.
Then  the conditional expectation of $\v_i$ observing $[u_i,v_i,\delta_i,\gamma_i]$ is
\begin{align}
&\E(\v \mid u_i-\hat{f}(x_i), v_i-\hat{f}(x_i), \delta_i, \gamma_i,x_i) \nonumber\\
&= \frac{\delta \G(\tilde{u}) + \gamma \{\G(\tilde{v})-\G(\tilde{u})\} + (1-\delta-\gamma) \{1-\G(\tilde{v})\}}{\delta \Fv(\tilde{u}) + \gamma \{\Fv(\tilde{v})-\Fv(\tilde{u})\} + (1-\delta-\gamma) \{1-\Fv(\tilde{v})\}} \nonumber \\
&= \frac{(\delta - \gamma) \G(\tilde{u}) + (\delta + 2\gamma - 1 ) \G(\tilde{v}) + (1-\delta-\gamma)}{(\delta - \gamma) \Fv(\tilde{u}) + (\delta + 2\gamma - 1 ) \Fv(\tilde{v}) + (1-\delta-\gamma)}	\label{eq23}
\end{align}
where $\tilde{u}=u_i - \hat{f}(x_i)$, $\tilde{v}=v_i - \hat{f}(x_i)$.
The above formula (\ref{eq23}) enables matrix calculations in standard software such as \textit{R} and \textit{Python} to accelerate the implementation.

If $\v$ follows a Gaussian distribution, $\G(\cdot)$ is in a closed form, and $\Fv(\cdot)$ may be approximated using Mills inequality:
$$
\frac{f(x)}{\sqrt{2+x^2}} \leq  P(\v > x) \leq \frac{f(x)}{\sqrt{2/\pi+x^2}}, \quad \forall z > 0.
$$
where $f(\cdot)$ is the density function of standard Gaussian.
We suggest
\begin{align*}
	&\Fv(s) = \frac{f(z)}{\sqrt{2+z^2}} ,\quad
	\G(s) = -\frac{e^{-s^2/2}}{\sqrt{2\pi}}
\end{align*}
when $U,V$ have large variances compared with $Y$ (so that the above approximation is tight), and numerical computation otherwise. 

Through experimental studies, we found that the results are not sensitive to the specified distribution of $\v$, e.g., a Logistic distribution. 
In practice, we may simply assume that $\v$ follows the standard Logistic distribution for computational convenience. 
In particular, Equation (\ref{eq23}) can be calculated in a closed form with 
\begin{align*}
	&\Fv(s) = \frac{1}{1+e^{-s}} , \\
	&\G(s) = \int_{-\infty}^s \frac{x e^x}{(1+e^x)^2} dx
	= \int_{(1+e^s)^{-1}}^{1} \{\log (1-t) - \log (t)\} dx 
	= - H\biggl(\frac{1}{1+e^{-s}}\biggr) 
\end{align*}
where $H: z \mapsto -z\log z - (1-z)\log (1-z)$ is the binary entropy function. 
For a general Logistic noise with zero mean and $\sigma$ standard deviation, the above $\Fv(s)$ and $G_\v(s)$ are replaced with $\Fv(s/\sigma)$ and $\sigma G_\v(s/\sigma)$, respectively.

\section*{S3. Further Discussions on Related Work} 

\subsection*{S3.1 Local differential privacy and its relationship with interval privacy}

A popular notation of privacy is the following local differential privacy~(see, e.g.,~\cite{evfimievski2003limiting,kasiviswanathan2011can,sarwate2014rate}).
\begin{definition}[Local Differential Privacy]
	For a given privacy parameter $\alpha \geq 0$, a privacy mechanism $\M$ is $\alpha$-differentially locally private if for all $y_1,y_2 \in \Y$,
	\begin{align}
		\sup_{A \in \sigma(\Z)} \frac{\P_{Z \mid Y}(z \in A \mid Y=y_1)}{\P_{Z \mid Y}(z \in A \mid Y=y_2)} \leq e^{\alpha} \label{eq31}	
	\end{align}
where $\sigma(\Z)$ denotes an appropriate $\sigma$-field over $\Z$.
\end{definition}

Both the above privacy and interval privacy are local, suitable for scenarios where data collecting agents are untrustworthy.  
When the conditional densities exist, an equivalent condition of (\ref{eq31}) is to require
\begin{align}
\frac{p_{Z \mid Y}(z  \mid Y=y_1)}{p_{Z \mid Y}(z \mid Y=y_2)} \leq e^{\alpha} \label{eq32}
\end{align}
for all $z\in \Z$ and $y_1,y_2 \in \Y$ (almost surely). 
Suppose that a joint distribution of $Y, Z$ exists.
By the Bayes' theorem, (\ref{eq32}) is further equivalent to
\begin{align}
	\frac{p_{Y \mid Z}(y_1  \mid Z=z)}{p_{Y \mid Z}(y_2  \mid Z=z)} \leq \frac{p_Y(y_1)}{p_Y(y_2)} 	e^{\alpha} . \label{eq33}
\end{align}
Compared with Definition~\ref{def_IP}, the requirement in (\ref{eq33}) holds for all $y_1,y_2$ but allows a flexibility of the likelihood ratio.

The following result shows that interval privacy and local differential privacy do not imply each other, and their intersection is a trivial solution with null utility and maximal privacy (or $\tau = 1$ and $\alpha = 0$).

\begin{proposition}[Intersection of Interval Privacy and Local Differential Privacy] \label{prop_intersect}
	A privacy mechanism $\M: Y \mapsto Z$ that simultaneously satisfies $\tau$-interval privacy and $\alpha$-local differential privacy ($\alpha<\infty$) is trivial, meaning that $Z$ and $Y$ have to be independent. 
\end{proposition}

As we mentioned in the main paper, another related notion of privacy is $\alpha$-information privacy~\cite{du2012privacy,sun2016towards} that requires the posterior-prior density ratio $p_{Y \mid Z}(y \mid Z=z)/p_Y(y)$ to be within $[e^{-\alpha},e^{\alpha}]$ for all feasible $y$ and $z$ and for a constant $\alpha>0$. By its definition, $\alpha$-information privacy implies  $2\alpha$-local differential privacy. Consequently, Proposition~\ref{prop_intersect} implies that interval privacy and information privacy do not imply each other. 

An interesting problem is to relate interval privacy and local differential privacy quantitatively.
Though interval privacy is neither weaker nor stronger than $\alpha$-differential privacy, a possible way of relating these two is through privacy-utility tradeoffs. 
Specifically, we first record the privacy-accuracy tradeoff curve under each privacy framework and then map the two parameters ($\tau$ and $\alpha$) through the same learning performance on the curve. The above will provide a way to define `analogous parameters' for interpretation and perception mathematically.  
In a numerical example, we generate $100$ i.i.d. samples of $Y\sim \textrm{Uniform}[0,1]$ and suppose that the distribution of $Y$ is unknown except that it falls into $[0,1]$. We applied Case-I interval mechanism with $U \sim \textrm{Uniform}[-b,1+b]$ with $b=[22, 8, 5, 4, 2, 1.5, 1]$.
We applied the technique in Example~\ref{eg2} to estimate $\mu \de \E(Y)$. We measure the utility as $\E|\hat{\mu}-\mu|$, where $\E$ is approximated from $1000$ independent replications. 
For comparison, we also used the $\alpha$-local differential privacy mechanism by perturbing $Y$ with Laplacian noises. We choose $\alpha=0.2,0.5,0.8,1,2,2.5,3$ so that the utility under each $\alpha$ is almost the same as that under the counterpart $b$ of interval privacy. We visualize the `analogous parameters' in Fig.~\ref{fig_tradeoff}. 
We note that the above example is only for illustration purposes. In general, there exists no universal relationship between $\tau$ and $\alpha$, as the tradeoff curves depend on the underlying learning task and privacy mechanisms. An interesting future direction is to use human perception (of privacy) as an evaluation criterion additionally to mathematical quantities such as $\tau$ and $\alpha$.

\begin{figure}[tb]
\begin{center}
\centerline{\includegraphics[width=0.5\columnwidth]{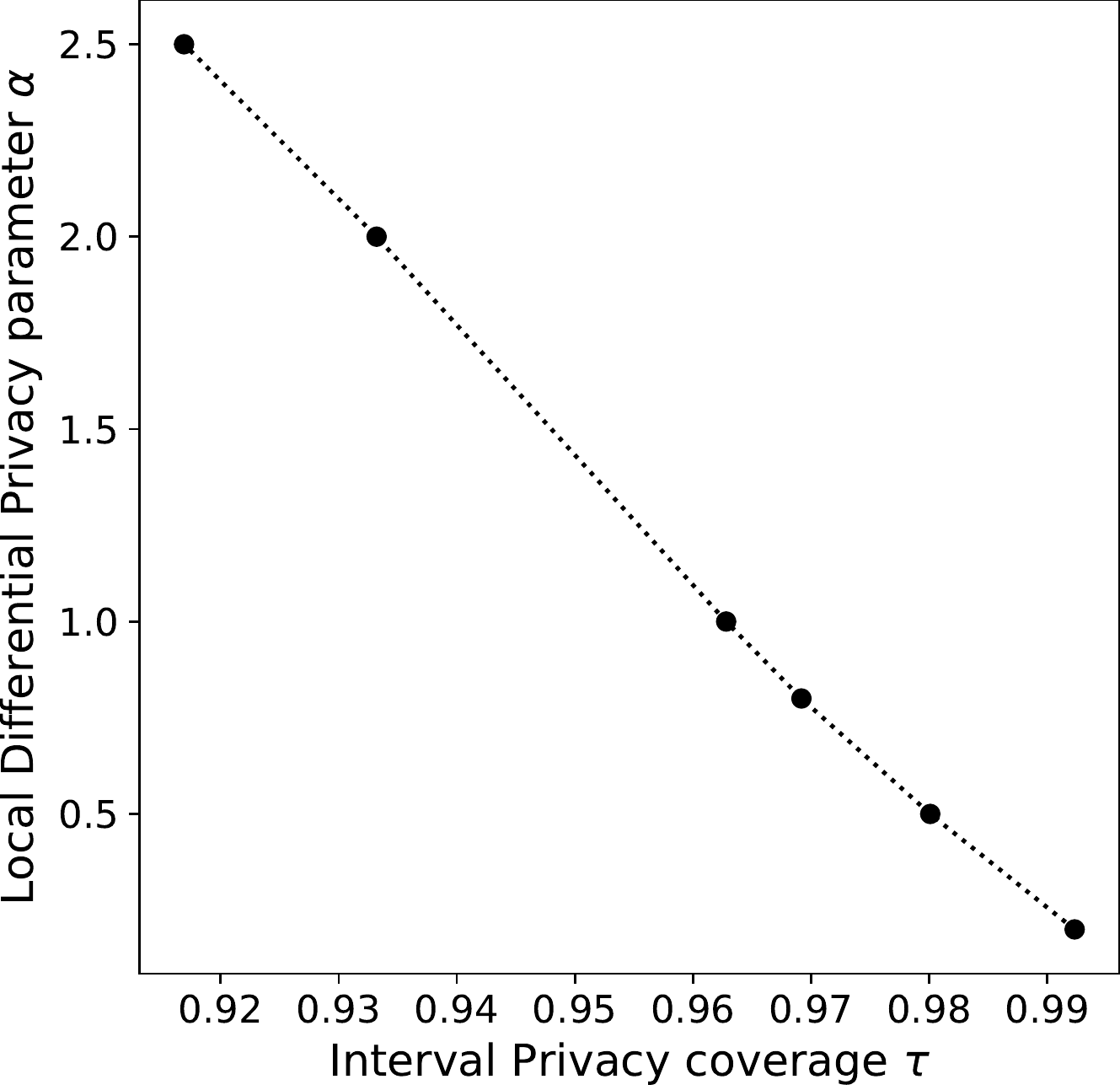}}
\vskip -0.1in
\caption{Illustration of the analogy between the $\tau$ of interval privacy and $\alpha$ of local differential privacy, linked through the same utilities on a particular mean estimation task. }
\label{fig_tradeoff}
\end{center}
\vskip -0.2in
\end{figure}

\subsection*{S3.2 Proof of Proposition~\ref{prop_intersect}}

	Suppose that a privacy mechanism $\M$ satisfies both interval privacy and local differential privacy.
	Recall that $S_z$ is the support of $Y$ given $Z=z$, and that $\Y$ is the support of $Y$. 
	
	We first show that $S_z \equiv \Y$ by contradiction. 
	Assume that $S_z \neq \Y$ for some feasible $z$, then $S_z \subsetneq \Y$ and there exists $y_1 \in S_z, y_2 \not\in S_z$. Consequently, the left-hand side of (\ref{eq33}) is infinity, violating the requirement of Inequality (\ref{eq33}) and thus local differential privacy.
	
	Therefore, $S_z = \Y$ for (almost surely) all $z$ in Definition~\ref{def_IP}, which further implies the independence of $Y$ and $Z$.

\subsection*{S3.3 Interval-Differential Privacy: a generalization of both worlds}
 
Motivated by the form of interval privacy and local differential privacy (\ref{eq33}), we introduce the following generalization. 
 
\begin{definition}[Interval-Differential Privacy]\label{def2} 
The same as Definition~\ref{def_IP} except that the second condition (\ref{eq_def}) is replaced with
	\begin{align}
		\frac{p_{Y \mid Z}(y_1 \mid Z=z)}{p_{Y \mid Z}(y_2 \mid Z=z)}	
		\leq \frac{p_{Y}(y_1)}{p_{Y}(y_2)} e^{\alpha}\label{eq_def2}
	\end{align}
	for $y_1,y_2$ in a range determined by $z$.
	A mechanism $\M$ is said to meet $(\tau,\alpha)$-interval-differential privacy if $\tau(\M)\geq \tau$.
\end{definition}

Compared with Definition~\ref{def_IP}, the requirement in (\ref{eq_def2}) is weaker as it involves a (typically small) parameter $\alpha$.
In the interval-differential privacy, any two conditional densities are only required to be equivocal for significant coverage of $y$ instead of all the support.
Thus, the $(\tau,\alpha)$-interval-differential privacy generalizes both local differential privacy (which corresponds to $\alpha=0$) and interval privacy (which corresponds to $\tau=1$).
A mechanism to realize the interval-differential privacy is first to perturb the raw data $Y$ to $\tilde{Y}$ and then report a random interval that contains $\tilde{Y}$.
In principle, this relaxation will enable a more flexible design of private data collecting and learning procedures. Further study of such a generalization is left as future research.

\subsection*{S3.4 Remark on the terms `obfuscation' and `perturbation'}

Here, we elaborate on the subtle difference between {obfuscation} and {perturbation}, referred to in Subsection~\ref{subsec_IP_def} of the main paper. 
We use the term `obfuscation' to refer to the process $Y \rightarrow Z$ that any deductive reasoning based on $Z$ does not contradict the truth of $Y$, referred to as information fidelity. For example, given $Z=[60,\i_{Y\leq 60}]$ that represents an interval $S_Z=(-\infty,60]$, one may conclude that $Y \leq b$ for any $b\geq 60$. In contrast, `perturbation' means one cannot make a factual statement from observing $Z$. To further clarify, we provide a mathematical description of the two terms below. 

Following the same notation in the paper, let us consider a general mechanism $\M: Y \mapsto Z$, where the raw data $Y$ and privatized data $Z$ are in the ranges $\Y$ and $\Z$, respectively. For any $z \in \Z$, let $\M^{-1}(z) \subseteq \Y$ denote the set of feasible $y$ (also called preimage) that can be mapped to $z$. Let $\M^{-1}(\Z)=\{\M^{-1}(z): z \in \Z\}-\{\emptyset, \Y\}$ denote those nontrivial preimages.

\textit{Obfuscation}: For every set $\YY \in \M^{-1}(\Z)$, there exists a non-empty set $\ZZ \in \Z$, such that $z \in \ZZ$ implies $y \in \YY$, namely
\begin{align}
	\forall \YY \in \M^{-1}(\Z), \quad \exists \ZZ \subseteq \Z, \ZZ \neq \emptyset, \quad s.t. \quad (z \in \ZZ) \Longrightarrow (y \in \YY) . \label{eq_obf}
\end{align}
In other words, a factual statement regarding $Y$ (abstracted by $y \in \YY$) can be possibly implied by a deduction based on $Z$ (abstracted by $z \in \ZZ$). The above (\ref{eq_obf}) holds for interval mechanisms. In fact, every $\YY \in \M^{-1}(\Z)$ is an interval/range $S_z$ (following the same notion in the paper) that is determined by a $z \in \Z$. So if we simply let $\ZZ = \{z\}$, we will have $(z \in \ZZ) \Longrightarrow (y \in \YY)$.
Note that the privacy is materialized by probabilistic $\YY$ and its width.

\textit{Perturbation}: For every set $\YY \in \M^{-1}(\Z)$ and every non-empty set $\ZZ \in \Z$, $z \in \ZZ$ does not imply $y \in \YY$, namely
\begin{align}
	\forall \YY \in \M^{-1}(\Z), \quad \forall \ZZ \subseteq \Z, \ZZ \neq \emptyset , \quad (z \in \ZZ) \not\Longrightarrow (y \in \YY) . \label{eq_per}
\end{align}
From the above definitions, it is clear that the characteristics of obfuscation and perturbation are mutually exclusive. We show that (\ref{eq_per}) holds for local differential privacy. In fact, if (\ref{eq_per}) does not hold, there exist sets $\YY \in \M^{-1}(\Z)$ and $\ZZ \in \Z$ such that $(z \in \ZZ) \Longrightarrow (y \in \YY)$. Since $\YY \not\in \{\emptyset, \Y\}$, for every $z \in \ZZ$, let $y_1 \in \YY$ denote the associated raw data and pick up any $y_2 \neq \YY$. Then, the density ratio $p_{Z \mid Y}(z \mid y_1)/p_{Z \mid Y}(z \mid y_2)$ is unbounded, violating the requirement of local differential privacy.

\section*{S4. Proof of Lemma~\ref{lemma1}}

	Let the two terms in the left-hand side be $A_1,A_2$, and the term in the right-hand side be $A_3$. We need to prove $A_1+A_2 \geq A_3$.
	Calculations show that
	\begin{align*}	
		A_1 
		& = 1 - \sum_{k=1}^K \biggl(\sum_{i \in S_k} x_i\biggr)^2 
		=\biggl(\sum_{k=1}^K \sum_{i \in S_k} x_i\biggr)^2 - \sum_{k=1}^K \biggl(\sum_{i \in S_k} x_i\biggr)^2 \\
		&= 2 \sum_{1\leq k\neq k' \leq K} \biggl(\sum_{i \in S_k} x_i\biggr)\biggl(\sum_{i' \in S_{k'}} x_{i'}\biggr).
	\end{align*}
	Similarly, we obtain
	\begin{align*}
		A_2 = 2 \sum_{1\leq j\neq j' \leq J} \biggl(\sum_{i \in S_j} x_j\biggr)\biggl(\sum_{i' \in S_{j'}} x_{j'}\biggr).
	\end{align*}
	It remains to prove that
	\begin{align*}
		2 \sum_{1\leq k\neq k' \leq K} \biggl(\sum_{i \in S_k} x_i\biggr)\biggl(\sum_{i' \in S_{k'}} x_{i'}\biggr) + 	2 \sum_{1\leq j\neq j' \leq J} \biggl(\sum_{i \in S_j} x_j\biggr)\biggl(\sum_{i' \in S_{j'}} x_{j'}\biggr)
		\geq 2 \sum_{1\leq i\neq i'\leq m} x_i x_{i'}.
	\end{align*}
	Without loss of generality, we consider a particular $i\in \Omega$, which belongs to $S_{k_*}$ for some unique $1\leq k_* \leq K$, and also belongs to $R_{j_*}$ for some unique $1\leq j_* \leq J$.
	We only need to prove the terms in $A_1+A_2$ that involve $x_i$ is no less than the corresponding terms in $A_3$.

	The terms in $A_1+A_2$ that involve $x_i$ are
	\begin{align}
		2\sum_{k \neq k_*}x_i\sum_{i' \in S_{k}} x_{i'}+ 
		2\sum_{j \neq j_*}x_i\sum_{i' \in R_{j}} x_{i'}
		=2x_i \biggl(\sum_{i' \in \Omega-S_{k_*}}	x_{i'}
		+\sum_{i' \in \Omega-S_{j_*}} x_{i'}\biggr)\label{eq_51}
	\end{align}
	where  the minus in $\Omega-S_{j_*}$ denotes the set difference, namely indices in $ \Omega$ but not in $S_{j_*}$. 
	Next, we prove that 
	\begin{align}
		\Omega-\{i\} \subseteq (\Omega-S_{k_*})\cup (\Omega-R_{j_*}). \label{eq_50}
	\end{align}
	In fact, by the assumption that the intersection of $S_k,R_j$ for any $k,j$ contains at most one element, for any $i'\in \Omega-\{i\} $, we must have $i' \in \Omega-S_{k_*}$ or $i'\in \Omega-R_{j_*}$. This implies (\ref{eq_50}).  
	It further follows from (\ref{eq_50}) that (\ref{eq_51}) is no larger than 
	$$
		2x_i \sum_{i' \in \Omega-\{i\}}	x_{i'},
	$$
	which concludes the proof.

\section*{S5. Proof of Lemma~\ref{lemma_NPMLE}}
 
Fix $\delta \in (0,1)$, $L=1+\lfloor 1/\delta^2\rfloor$, and let $-\infty=t_0^{(j)}<t_1^{(j)}<\cdots t_H^{(j)}=\infty$ (for each $j\in [1:q]$) be a grid of points such that
$\int_{t^{(j)}\in [t_{h-1}^{(j)}, t_{h}^{(j)}]} dt^{(j)} dt^{(-j)} = 1/H$, $h \in [1:H]$, where $t^{(-j)}$ denotes the subvector of $t$ excluding the $j$-th entry. For $j \in [1:q]$, we let 
\begin{align}
	J_{h}^{(j)} &\de (t_{h-1}^{(j)}, t_{h}^{(j)}], \quad h \in [1:H], \nonumber \\
	I^{(j)} &\de \biggl\{ h \in [1:H]: \, G^{(j)} \circ F(t_{h}^{(j)})-G^{(j)} \circ F(t_{h-1}^{(j)}) \geq \delta \biggr\} , \nonumber\\
	\bar{I}^{(j)} &\de [1:H]-I^{(j)}.\label{eq_205}
\end{align}
Since $G^{(j)} \circ F$ is nondecreasing and bounded, the cardinality of $I^{(j)}$ satisfies
\begin{align}
	\card(I^{(j)}) = O(\delta^{-1})	, \quad j \in [1:q] . \label{eq_card}
\end{align}

Next, we will show that for each $i \in [1:m]$,
\begin{align}
	&\int_{A_{\v} \times \Y} \sum_{i=1}^m\B^{(i)} \frac{\Fy(\Ra_t^{(i)})}{\hat{F}_{n_k}(\Ra_t^{(i)}, \omega)} \, d \P_{n_k}(t,y) \label{eq_206}  \\
	&=\int_{A_{\v} \times \Y} \sum_{i=1}^m \B^{(i)} \frac{\Fy(\Ra_t^{(i)})}{\hat{F}_{n_k}(\Ra_t^{(i)}, \omega)} \, d \P(t,y) + r_k^{(i)}(\omega) \label{eq_203} \\
	&\textrm{ with } |r_k^{(i)}(\omega)| \leq c_i\delta \label{eq_208}
\end{align}
for constants $c_i>0$.
Note that the term in (\ref{eq_206}) may be written as
\begin{align}
	\sum_{h=1}^H \int_{(\R\times \cdots J_h^{(j)} \times \cdots \R \ \cap \ A_{\v} ) \times \Y} \sum_{i=1}^m\B^{(i)} \frac{\Fy(\Ra_t^{(i)})}{\hat{F}_{n_k}(\Ra_t^{(i)}, \omega)} \, d \P_{n_k}(t,y) \nonumber	
\end{align}
for all $j\in[1:q]$.

For a generic $t$, suppose that its $j$-th entry $t^{(j)}$ falls into the interval $J_{h}^{(j)}$.
By invoking the Monotonicity condition (a), the denominator in (\ref{eq_203}) satisfies 
\begin{align}
	\hat{F}_{n_k}(\Ra_t^{(i)}, \omega) 
	\leq \hat{F}_{n_k}(\Ra_{\tilde{t}}^{(i)}, \omega) , \label{eq_204}
\end{align}
where $\tilde{t}$ is the same as $t$ except that its $j$-th entry is $\tilde{t}^{(j)} \de t_{h}^{(j)}$ if $F(\Ra_t^{(i)})$ is non-decreasing in $t^{(j)}$, and $\tilde{t}^{(j)} \de t_{h-1}^{(j)}$ otherwise. 
Then, it follows from (\ref{eq_C}) and the Monotonicity condition (b) that 
\begin{align}
	\biggl|\frac{1}{\hat{F}_{n_k}(\Ra_t^{(i)}, \omega)} 
	 -\frac{1}{\hat{F}_{n_k}(\Ra_{\tilde{t}}^{(i)})}\biggr|
	 & \leq \frac{\bigl|\hat{F}_{n_k}(\Ra_{\tilde{t}}^{(i)})-\hat{F}_{n_k}(\Ra_{t}^{(i)})\bigr|}{\bigl(\hat{F}_{n_k}(\Ra_t^{(i)}, \omega) \bigl)^2} \nonumber \\
	 & \leq \frac{\bigl|G^{(j)} \circ \hat{F}_{n_k}(\tilde{t}^{(j)})-G^{(j)} \circ \hat{F}_{n_k}(t^{(j)})\bigr|}{\bigl(\hat{F}_{n_k}(\Ra_t^{(i)}, \omega) \bigl)^2} \nonumber \\
	 &\leq \bigl|G^{(j)} \circ \hat{F}_{n_k}(\tilde{t}^{(j)})-G^{(j)} \circ \hat{F}_{n_k}(t^{(j)})\bigr| \cdot C^2 . \label{eq_211} 
\end{align}
Because $G^{(j)}$ is continuous and $\hat{F}_{n_k}$ converges to $F$, it follows from (\ref{eq_211}) that for all $t$ with $t^{(j)} \in J_{h}^{(j)}$, $h \in \bar{I}^{(j)}$ (defined in (\ref{eq_205})) and all sufficiently large $k$, we have 
\begin{align}
	\biggl|\frac{1}{\hat{F}_{n_k}(\Ra_t^{(i)}, \omega)} 
	 -\frac{1}{\hat{F}_{n_k}(\Ra_{\tilde{t}}^{(i)})}\biggr|
	&\leq 2 \delta C^2 . \label{eq_207}
\end{align}
To bound the variation for $I^{(j)}$, $j\in [1:q]$, we use (\ref{eq_card}) and that $\P(\R\times \cdots J_h^{(j)} \times \cdots \R \times \Y)=O(\delta^2)$ to obtain
\begin{align}
	\sum_{h \in I^{(j)}} d\P(\R\times \cdots J_h^{(j)} \times \cdots \R \times \Y)=O(\delta)  \nonumber .
\end{align}
This, in conjunction with (\ref{eq_C}) and (\ref{eq_207}) imply the desired bound~(\ref{eq_208}).

By the dominated convergence theorem, we have
\begin{align}
	&\lim_{k \rightarrow \infty} \int_{A_{\v} \times \Y} \sum_{i=1}^m \B^{(i)} \frac{\Fy(\Ra_t^{(i)})}{\hat{F}_{n_k}(\Ra_t^{(i)})} \, d \P(t,y) \nonumber \\
	&=\int_{A_{\v} \times \Y} \sum_{i=1}^m \B^{(i)} \frac{\Fy(\Ra_t^{(i)})}{F(\Ra_t^{(i)})} \, d \P(t,y) \label{eq_209} .
\end{align}	
Combining (\ref{eq_208}) and (\ref{eq_209}), we obtain
\begin{align}
	&\int_{A_{\v} \times \Y} \sum_{i=1}^m\B^{(i)} \frac{\Fy(\Ra_t^{(i)})}{\hat{F}_{n_k}(\Ra_t^{(i)}, \omega)} \, d \P_{n_k}(t,y) \nonumber \\
	&=\int_{A_{\v} \times \Y} \sum_{i=1}^m \B^{(i)} \frac{\Fy(\Ra_t^{(i)})}{F(\Ra_t^{(i)})} \, d \P(t,y) + \tilde{r}_k^{(i)}(\omega) \nonumber 
\end{align}	
for all sufficiently large $k$, with $|\tilde{r}_k^{(i)}(\omega)|\leq \tilde{c}_i$ for constants $\tilde{c}_i>0$. Since $\delta$ can be chosen arbitrarily chosen, we conclude the proof of Lemma~\ref{lemma_NPMLE}.

\section*{S6. Distributional Identifiability of Selective Mechanisms}

In Subsection~\ref{subsec_individual_privacy}, we discussed a selective mechanism. 
We will show that although such an interface inevitably introduces a selective bias (towards the interval values), the NPMLE may still be asymptotically consistent. 
Our analyses will be based on the extended notion of interval mechanism in Definition~\ref{def_extended_mechanism}.

Recall that $\B^{(i)}(\T, Y) \de \i_{Y \in \Ra_{\T}^{(i)}}$ for $i\in [1:m]$. We write it as $\B^{(i)}$ when its dependency on $Y$ and $\T$ is clear from the context. 
We still consider the log-likelihood functional 
\begin{align}
	\psi: F \mapsto 
	&\int_{\R^q \times \Y} \sum_{i=1}^m \B^{(i)}(t,y) \log F(\Ra^{(i)}) \, d \P_n(t,y) , \label{eq_loglik2} 
\end{align}
where $\P_n(\cdot,\cdot)$ denotes the empirical probability measure from from $[\T_j,Y_j]$ ($j\in [1:n]$) that represent the collected data.
Let the NPMLE be a right-continuous distribution function that maximizes $\psi(F)$.
Note that although (\ref{eq_loglik2}) is in the same form as (\ref{eq_loglik}), the asymptotic limit of $\P_n$ may be different due to the dependence of $\T$ and $Y$. 

Note that the observed data are i.i.d. from the distribution of $d\P(t,y)$ restricted to the region that there exists at least a $t$-generated range that has coverage of at least $\tau$ and that $y$ belongs to one of those ranges.
	As such, we let
	\begin{align}
		\mathcal{T}_{\tau} \de \biggl\{t \in \R^{q} : \exists i \in [1:m] \textrm{ such that } L(R_t^{(i)}) \geq \tau \biggr\} \nonumber
	\end{align}
	denote the feasible set of $t$ in the sense that it admits at least one range of coverage at least $\tau$. Correspondingly, we let $\P_{\T \mid \tau}$ denote the probability of $\T$ conditional on $\T \in \mathcal{T}_{\tau}$.
We first consider a fixed $\tau \in [0,1)$ and introduce the following condition. 

\vspace{0.1cm}
\noindent \textbf{$\tau$-Resolvability condition}:
For any continuous CDF $\Fy$, 
\begin{align}	
	&\int_{\mathcal{T}_{\tau}} \frac{1}{\sum_{i: L(\Ra_t^{(i)}) \geq \tau} \Fy(\Ra_t^{(i)})} \cdot \sum_{i: L(\Ra_t^{(i)}) \geq \tau} \frac{\Fy^2(\Ra_t^{(i)})}{ F(\Ra_t^{(i)})} \, d \P_{\T \mid \tau}(t) 
	 \leq 1 \label{eq_cauchy2} 
\end{align}
implies that $F=\Fy$.	

To develop intuitions, let us consider $\tau=0$. We then have $L(\Ra_t^{(i)}) \geq \tau$ for all $i \in [1:m]$ and Inequality (\ref{eq_cauchy2}) becomes 
\begin{align}	
	&\int_{\mathcal{T}_{\tau}} \sum_{i=1}^m \frac{\Fy^2(\Ra_t^{(i)})}{ F(\Ra_t^{(i)})} \, d \P_{\T \mid \tau}(t) 
	 \leq 1 \nonumber.
\end{align}
It follows from the Cauchy's inequality that $\Fy(\Ra_t^{(i)})=F(\Ra_t^{(i)})$ for all $i$ and $t$ with a positive density. With the Resolvability condition, we can further derive $F=\Fy$. 
For a positive $\tau$, verifying the $\tau$-Resolvability condition is not straightforward. This condition is regarded as stronger than the previous Resolvability condition.

\begin{proposition} \label{thm_selective_NPMLE}
	Assume that an extended interval mechanism satisfies the $\tau$-Resolvability and Monotonicity conditions, and $\Fy$ is continuous. Then, $\sup_{y\in\Y}|\hat{F}_n(y) - \Fy(y)| \rightarrow 0$ almost surely as $n \rightarrow \infty$.
\end{proposition}

\begin{proof}[Proof of Proposition~\ref{thm_selective_NPMLE}]

We let
	\begin{align}
		\mathcal{Y}_{t} \de \biggl\{y \in \Y : \exists i \in [1:m] \textrm{ such that } y \in R_t^{(i)}, \ L(R_t^{(i)}) \geq \tau \biggr\} \nonumber
	\end{align}
	denote the feasible set of $y$ conditional on $T=t$.
	Note that $\mathcal{Y}_{t}$ is an empty set when $t \not\in \mathcal{T}_{\tau}$.
	Let $\P_{Y \mid t}$ denote the probability of $Y$ conditional on $T=t$.
	It can be seen that the probability of $\B^{(i)}(t, Y)=1$ conditional on $Y \in \mathcal{Y}_{t}$ for a given $t \in \mathcal{T}_{\tau}$ is
	\begin{align}
		\P_{t}\bigl(\B^{(i)}(t, Y)=1 \mid Y \in \mathcal{Y}_{t}\bigr) = \frac{\Fy(R_t^{(i)})}{\sum_{i': L(R_t^{(i')}) \geq \tau} \Fy(R_t^{(i')})} \textrm{ if $i$ satisfies $R_t^{(i')}\geq \tau$, and $0$ otherwise}. \label{eq_214}
	\end{align}

	Since $\hat{F}_n$ is an NPMLE, for each $\v \in(0,1)$, we have $\lim_{\v\rightarrow 0^+} \v^{-1} \psi((1-\v)\hat{F}_n + \v \Fy) - \psi(\hat{F}_n) \leq 0$, which implies that
	\begin{align}
		&\int_{\R^q \times \Y} \sum_{i=1}^m \frac{\B^{(i)}(t,y) \Fy(\Ra_t^{(i)})}{\hat{F}_n(\Ra_t^{(i)})} \, d \P_n(t,y) \leq 1 . \label{eq_212}
	\end{align}
	By a similar argument as in the proof of Theorem~\ref{thm_NPMLE}, we have a limiting counterpart of (\ref{eq_212}) as follows.
		\begin{align}
			&\int_{\mathcal{T}_{\tau}} \int_{\mathcal{Y}_{t}} \sum_{i=1}^m \frac{\B^{(i)}(t,y) \Fy(\Ra_t^{(i)})}{F(\Ra_t^{(i)})} \, d \P_{Y \mid t}(y) \, d \P_{\T \mid \tau}(t) 
			\leq 1. \label{eq_213}
		\end{align}	
	Meanwhile, using (\ref{eq_214}), we obtain for a fixed $t$ that
	\begin{align}
		\int_{\mathcal{Y}_{t}} \sum_{i=1}^m \frac{\B^{(i)}(t,y) \Fy(\Ra_t^{(i)})}{F(\Ra_t^{(i)})} \, d \P_{Y \mid t}(y) 
		&=\sum_{i=1}^m \int_{\mathcal{Y}_{t}} \frac{\B^{(i)}(t,y) \Fy(\Ra_t^{(i)})}{F(\Ra_t^{(i)})} \, d \P_{Y \mid t}(y) \\
		&= \sum_{i: L(R_t^{(i)}) \geq \tau} \frac{\Fy(R_t^{(i)})}{\sum_{i': L(R_t^{(i')}) \geq \tau} \Fy(R_t^{(i')})} \frac{\Fy(\Ra_t^{(i)})}{F(\Ra_t^{(i)})}\\
		&=\frac{1}{\sum_{i: L(\Ra_t^{(i)}) \geq \tau} \Fy(\Ra_t^{(i)})} \cdot \sum_{i: L(\Ra_t^{(i)}) \geq \tau} \frac{\Fy^2(\Ra_t^{(i)})}{ F(\Ra_t^{(i)})}. \label{eq_215}
	\end{align}
	Combining Inequality (\ref{eq_213}) and Equality (\ref{eq_215}), we obtain Inequality (\ref{eq_cauchy2}),
	Then, the $\tau$-Resolvability condition implies that $F=\Fy$. With a similar argument as in the proof of Theorem~\ref{thm_NPMLE}, we further conclude Proposition~\ref{thm_selective_NPMLE}.
\end{proof}

We can consider using a random $\tau \in [0,1)$ to represent different individual-level privacy sensitivities. Suppose that $\tau$ is independent with $[\T,Y]$. The above Proposition~\ref{thm_selective_NPMLE} can be directly extended by replacing Inequality (\ref{eq_cauchy2}) with 
\begin{align}	
	&\int_{[0,1)} \int_{\mathcal{T}_{\tau}} \frac{1}{\sum_{i: L(\Ra_t^{(i)}) \geq \tau} \Fy(\Ra_t^{(i)})} \cdot \sum_{i: L(\Ra_t^{(i)}) \geq \tau} \frac{\Fy^2(\Ra_t^{(i)})}{ F(\Ra_t^{(i)})} \, d \P_{\T \mid \tau}(t) d\P(\tau)
	 \leq 1 \nonumber .
\end{align}

\section*{S7. Additional Remarks}

\subsection*{S7.1 Remark on the threat model}

Here, we briefly summarize the threat model often considered in the study of local data privacy (and this paper). 

There are two kinds of participants during data collection. The first is a cohort of individuals, each holding a private data value abstracted by $Y$; The second is a data collector obtaining privatized data denoted by $Z$. The data collector may or may not be benign, unknown to the individuals. If the collector is benign, its goal is to perform population-level inference (regarding the distribution of $Y$); If it is adversarial, it aims to uncover a particular individual's identity or underlying value. 

In the local data privacy setting, individuals' data are not centralized in one place (since otherwise, it becomes a database privacy problem). Correspondingly, the data collection is separately operated for each individual. 
Meanwhile, it is standard to presume that the data collector has no identifier of any individual for two reasons. First, under emerging regulations on data sharing, a data collector often has to anonymize each individual by removing identifiable information during collection. Second, if the data collector can access an individual's identity, it can potentially link it to an external dataset that immediately exposes further information about this individual. In that scenario, little is guaranteed from interval privacy or local differential privacy, since existing notions of privacy are defined within the scope of a pre-determined dataset. Thus, in an entirely private data collection, identifiers need to be removed or privatized as a variable. Consequently, any side information is at a population level and can be taken into the prior distribution of $Y$ to evaluate the privacy coverage/leakage.

\subsection*{S7.2 Remark on the interval mechanisms for multi-dimensional data}

In the main paper, we assumed the raw data $Y$ to be a scalar for technical convenience. We can easily extend the related mechanisms and arguments to the multi-dimensional case.
Here, we mention three ways to use interval mechanisms for multi-dimensional data. 
One way is to privatize each data dimension separately. We have implemented this in our data experiments, where the survey asked multiple questions to the same individual. 
The second way is to treat the interval mechanism $R: t \mapsto \{R_t^{(i)}\}_{i=1}^m$ as a mapping from $t$ into a partition of a multi-dimensional set. For example, each $R_t^{(i)}$ can be regarded as a region on a geographic map. 
The third way is to privatize a function of the original data, say $f(\bm Y)$ instead of $\bm Y$. For a generated range $r$, it induces a range/preimage $f^{-1}(r)$ on the domain of $Y$. This way may be suitable for cases where $\bm Y$ is high-dimensional and can be represented by $f(\bm Y)$, e.g., a generative latent code (in the encoded space).

\end{document}